\documentclass[a4paper,11pt]{article}
\usepackage{amsmath, amssymb}
\usepackage{amsthm}
\usepackage[OT1]{fontenc}
\usepackage{url}
\usepackage[margin=1.1in]{geometry}
\usepackage{color}
\DeclareMathOperator*{\argmin}{arg\,min}

\newtheorem{theorem}{Theorem}

\newtheorem{lemma}[theorem]{Lemma}

\newtheorem{remark}[theorem]{Remark}

\newtheorem{assumption}[theorem]{Assumption}
\numberwithin{equation}{section}
\newcommand{\sign}{\text{sign}}
\title{  Constrained Quadratic Risk Minimization via Forward and Backward Stochastic Differential Equations}
\date{}
\author{Yusong Li\footnote{Department of Mathematics, Imperial College, London SW7 2BZ, UK.  Email: y.li11@imperial.ac.uk}\; and Harry Zheng\footnote{Department of Mathematics, Imperial College, London SW7 2BZ, UK.  Email: h.zheng@imperial.ac.uk}}

\begin{document}

\maketitle

\begin{abstract}
In this paper we study a continuous-time stochastic linear quadratic control problem arising from mathematical finance. We model the asset dynamics with  random market coefficients and portfolio strategies with convex constraints. 
Following the convex duality approach, we  show that the necessary and sufficient optimality conditions for both the primal and dual problems can be written in terms of processes satisfying a system of FBSDEs together with other conditions. We  characterise explicitly the optimal wealth and portfolio processes as  functions of adjoint processes from the dual FBSDEs in a dynamic fashion and vice versa. 
We apply the results  to solve quadratic risk minimization problems with  cone-constraints and derive the explicit representations of solutions to the extended stochastic Riccati equations  for such problems.
\end{abstract}



\section{Introduction}
In this paper we study a stochastic control problem arising from mathematical finance. The goal is to minimize a convex   cost function that is quadratic in both the wealth process and portfolio strategy in a continuous time complete market with random market parameters and portfolio constraints. Problems of this kind arise  naturally in financial applications. We assume that the portfolio must take value in a given closed convex set which is general enough to model  short selling, borrowing,   and other trading restrictions, see \cite{karatzasshreve:mathfinance}.

There are vast literatures on stochastic linear quadratic (SLQ) optimal control and its applications on mean variance portfolio selection problems, see \cite{schweizer:mvhedging, yong.zhou:stochasticcontrols}  and references therein.  In the absence of portfolio constraints, using the stochastic maximum principle, one can solve the SLQ problem by  deriving the optimal control as a linear feedback control of the state and proving the existence and uniqueness of a solution to the resulting stochastic Riccati equation (SRE). When there are no control constraints,  the feedback control constructed from the solution of the SRE is automatically admissible, see \cite{zhoulim:MVrandom} for an example of this method to problems with random coefficients but no portfolio constraints. When there are control constraints, the optimal control is no longer a simple linear feedback control of the state and   the SRE method becomes much more difficult and  subtle.  
\cite{huzhou:constrainedmv} shows the solvability of  an extended SRE for constrained SLQ problems with random coefficients.

For convex SLQ problems, it is also natural to use the convex duality method that has been extensively applied to solve utility maximization problems in mathematical finance, see \cite{KS99, KS03} and reference therein. When there are no control constraints and the filtration is generated by driving Brownian motions, one may first convert the original  dynamic optimization problem into an equivalent static one, then formulate and solve the static dual problem, and use the dual relation and the martingale property to find the optimal state process for the original problem, finally use the martingale representation theorem to find a replicating portfolio which is the optimal control process. When there are control constraints, the duality method becomes much more complicated. \cite{karatzasshreve:mathfinance} introduces and solves a family of auxiliary unconstrained  problems and shows one of them solves the original constrained problem.
\cite{heunislabbe:constrainedMV} applies the convex duality approach,  inspired by \cite{bismut:convexdual, rogers:constrainteddual}, to solve a mean-variance problem with both random coefficients and portfolio constraints and shows the existence of an optimal solution to the dual problem and constructs the optimal wealth process with the optimal dual solution and the optimal portfolio process with the martingale representation theorem. \cite{Czichowsky:convexdualconstraints} provides a comprehensive treatment on mean-variance hedging under convex trading constraints in a general semi-martingale setting. It establishes the closedness of the set of all replicable terminal wealth under trading constraints in some square integrable sense
and subsequently the existence of a solution to mean-variance hedging problems, and extends results linking the primal and dual problems obtained previously by other authors, see detailed discussions in \cite[Section 5.3]{Czichowsky:convexdualconstraints}.

\cite{oksendal:robustduality} extends the results of \cite{KS03} to a dynamic setting and proves a close relation between optimal solutions and adjoint processes obtained from forward backward  stochastic differential equations (FBSDEs). Specifically, it is shown that 
the optimal primal wealth and portfolio processes can be expressed as functions of  the optimal adjoint processes of the dual problem and vice versa. This demystifies the opaque relation of the optimal solutions of the primal and dual problems in utility maximization, i.e.,  given the solution of the dual problem, the optimal control of the primal problem can only be derived from the martingale representation theorem. There are no control constraints in \cite{oksendal:robustduality} but the asset price process is a general semi-martingale process with some technical conditions.

Inspired by the work of \cite{oksendal:robustduality}, we use the convex duality method to solve the quadratic risk minimization problem with both random coefficients and control constraints. To get a correct formulation of the dual problem, we follow the approach of \cite{heunislabbe:constrainedMV}  by first converting the original problem into a static problem in an abstract space, then applying convex analysis to derive its dual problem, and finally getting a specific dual stochastic control problem. It turns out there are three controls in the dual problem, one corresponds to the control constraint set, one to the running cost function, and one to the no-duality-gap relation. Using  FBSDEs, we derive the  necessary and sufficient conditions for both primal and dual problems, which allows us to explicitly characterise the primal control as a function of the adjoint process coming from the dual FBSDEs in a dynamic fashion and vice versa, similar to those in 
\cite{oksendal:robustduality}.  Moreover, we also find that the optimal primal wealth process coincides with the optimal adjoint process of the dual problem and vice versa.
To the best of our knowledge, this is the first time the dynamic relations of  primal and dual problems with control constraints have been explicitly characterized  in terms of solutions to their corresponding FBSDE systems.

 After establishing the optimality conditions for both primal and dual problems, we solve a quadratic risk minimization problem with cone-constraints.  Instead of attacking the primal problem directly, we start from the dual problem and then  construct the optimal solution to the primal problem from that of the dual problem. Moreover, we  derive the explicit representations of solutions to the extended SREs introduced in  \cite{huzhou:constrainedmv}   in terms of the optimal solutions  from the dual problem. The simplicity in solving the dual problem is in good contrast to the technical complexity  in solving the extended SREs directly, as discussed in  \cite{huzhou:constrainedmv}. In addition,  we show that when the coefficients are deterministic, the closed form optimal solution to the dual problem can be constructed.

The rest of the paper is organised as follows. In Section 2 we   set up the model and formulate the primal and dual problems following the approach in \cite{heunislabbe:constrainedMV}. In Section 3 we  characterise the necessary and sufficient optimality conditions for both the primal and dual problems and   establish their connection  in a dynamic fashion through  FBSDEs.  In Section 4 we discuss quadratic risk minimization problems with cone constraints and demonstrate how to construct explicitly the solutions of the extended SREs from those of the dual FBSDEs. 
In Section 5 we prove the main results. Section 6 concludes.

\section{Market Model and Primal and Dual Problems}
Through out the paper, we denote by $ T>0$ a fixed terminal time,  $\lbrace W(t),t\in[0,T] \rbrace$ a  $ \mathbb{R}^N $-valued standard Brownian motion with scalar entries $ W_m(t)$, $m=1,\cdots,N $, on a complete probability space $ (\Omega,\mathcal{F},\mathbb{P}) $, 
$\{\mathcal{F}_t\}$ the $\mathbb{P}$-augmentation of the filtration $\mathcal{F}_t^W=\sigma( W(s),0\leq s\leq t ) $ generated by $W$, 
 ${\cal P}(0,T; \mathbb{R}^N)$ the set of all $\mathbb{R}^N$-valued progressively measurable processes on $[0,T]\times \Omega$, ${\cal H}^2(0,T; \mathbb{R}^N)$ the set of processes $x$ in ${\cal P}(0,T; \mathbb{R}^N)$ satisfying 
$E[\int_0^T |x(t)|^2dt]<\infty$, and ${\cal S}^2(0,T; \mathbb{R}^N)$ the set of processes $x$ in ${\cal P}(0,T; \mathbb{R}^N)$ satisfying 
$E[\sup_{0\leq t\leq T} |x_t^2|]<\infty$. 
We write SDE for stochastic differential equation, BSDE for backward SDE, and FBSDE for forward and backward SDE. We also follow the customary convention that $\omega$ is suppressed in SDEs and integrals, except in  places where an explicit $\omega$ is needed.

Consider a market consisting of a bank account with price $ \lbrace S_0(t) \rbrace $ given by
\begin{equation}\label{bank_account}
dS_0(t) = r(t)S_0(t)dt, \ 0\leq t\leq T, \ S_0(0)=1,
\end{equation}
and $ N $ stocks with prices $ \lbrace S_n(t) \rbrace, \ n=1,\cdots,N $, given by 
\begin{equation}\label{stock_price}
dS_n(t) = S_n(t)\left[ b_n(t)dt+\sum_{m=1}^N\sigma_{nm}(t)dW_{m}(t) \right],\ 0\leq t\leq T, \ S_n(0)>0.
\end{equation}
We assume that 
$r \in {\cal P}(0,T; \mathbb{R})$ (scalar interest rate), $b\in {\cal P}(0,T; \mathbb{R}^N)$ (vector of appreciation rates), and $\sigma \in {\cal P}(0,T; \mathbb{R}^{N\times N})$ (volatility matrix) 
 are uniformly bounded. We also assume that there exists a positive constant $ k$ such that 
\begin{equation*}
z'\sigma(t)\sigma'(t)z\geq k|z|^2
\end{equation*}
for all $ (z,\omega,t)\in \mathbb{R}^N\times \Omega\times [0,T] $, where $z'$ is the transpose of $z$. The strong non-degeneracy condition above ensures that matrices $\sigma(t), \sigma' (t)$ are  invertible and uniformly bounded.  

Consider a small investor with initial wealth $ x_0>0 $ and a self-financing strategy.  Define the set of admissible portfolio strategies by
\begin{equation*}
\left.
\begin{array}{c}
\mathcal{A}:= \left\{ \pi \in  {\cal H}^2(0,T; \mathbb{R}^N): \pi(t)\in K  \textit{ for $t\in[0,T]$ a.e.}\right\},
\end{array}
\right.
\end{equation*}
where $ K \subseteq \mathbb{R}^N $ is a closed convex set containing $0$ and  $ \pi$ is a  portfolio process with   each entry $ \pi_n(t) $ defined as the amount  invested in the stock $n$ for $n=1,\ldots,N$. 
Given any  $ \pi\in\mathcal{A} $, the investor's total wealth $ X^{\pi}$ satisfies the SDE
\begin{equation}\label{wealth_process}
\left\{
\begin{array}{l}
dX^{\pi}(t)=[ r(t)X^{\pi}(t)+\pi'(t)\sigma(t)\theta(t) ] dt + \pi'(t)\sigma(t)dW(t),\quad 0\leq t\leq T,
\\
X^{\pi}(0)=x_0,
\end{array}
\right.
\end{equation}
where $\theta(t) := \sigma^{-1}(t)\left[ b(t)-r(t)\textbf{1} \right]$ is
the market price of risk at time $t$ and is uniformly bounded and  $ \textbf{1}\in\mathbb{R}^N $ has all unit entries.   A pair $(X,\pi)$ is  {\it admissible} if $\pi\in {\cal A}$ and $X$ is a strong solution to the SDE (\ref{wealth_process}) with control process $\pi$.

Define a functional $ J:\mathcal{A}\rightarrow\mathbb{R} $ by
\begin{equation*}
J(\pi):= E\left[ \int_0^T f(t,X^{\pi}(t),\pi(t))dt + g(X^{\pi}(T)) \right],
\end{equation*}
where $ f: \Omega\times [0,T]\times \mathbb{R}\times \mathbb{R}^N\rightarrow \mathbb{R} $ and $ g:\Omega\times \mathbb{R}\rightarrow\mathbb{R} $ are defined by 
\begin{equation}\label{cost_functions_f_g}
\left\{
\begin{array}{l}
f(\omega, t,x,\pi):= \dfrac{1}{2}\left[Q( t)x^2+2S'( t)x\pi + \pi' R( t)\pi\right],\vspace{2mm}\\\
g(\omega, x):= \dfrac{1}{2}\left[ a x^2 + 2c x \right].
\end{array}
\right.
\end{equation}
We assume that  random variables 
$a,c\in L^{\infty}_{\mathcal{F}_T}(\mathbb{R})$ satisfy
$$ 0< \inf_{\omega\in\Omega} a(\omega)\leq \sup_{\omega\in\Omega}a(\omega)<\infty$$
and processes  
$Q\in {\cal P}(0,T; \mathbb{R}),
S\in {\cal P}(0,T; \mathbb{R}^N), 
R\in {\cal P}(0,T; \mathbb{R}^{N\times N})$ are uniformly bounded, $R(t)$ is a symmetric matrix, and the matrix 
$$\left(\begin{array}{cc}
Q(t) & S'(t)\\ S(t) & R(t)
\end{array}\right)
$$
is non-negative definite for all $(\omega,t)\in  \Omega\times [0,T]$. 
Under these assumptions    we know $ J $ is a convex functional of $\pi$. 

We consider the following optimization problem:
\begin{equation}\label{primal_ori}
\mbox{Minimize $J(\pi)$ 
subject to $( X,\pi)$   admissible.}
\end{equation}
An admissible control $ \hat{\pi} $ is  {\it optimal} if 
$J(\hat{\pi})\leq J(\pi)$ for all $\pi\in\mathcal{A}$.

Following the approach introduced in \cite{heunislabbe:constrainedMV}, we now set up the dual problem. Denote by
\begin{flalign*}
\mathbb{B}:= &\mathbb{R}\times {\cal H}^2(0,T; \mathbb{R})\times {\cal H}^2(0,T; \mathbb{R}^N).
\end{flalign*}
We write $ X\in\mathbb{B} $ if and only if 
\begin{equation*}
X(t)=x_0+\int_0^t \dot{X}(\tau)d\tau + \int_0^t \Lambda_X'(\tau) dW(\tau), \; 0\leq t\leq T,
\end{equation*}
for some $ (x_0,\dot{X},\Lambda_X)\in\mathbb{B} $.
We now reformulate \eqref{primal_ori} as a primal optimization problem over the whole set $ \mathbb{B} $. For each $ X\equiv (x_0,\dot{X},\Lambda_X)\in\mathbb{B} $, define
\begin{flalign*}
\mathcal{U}(X):= \lbrace \pi\in\mathcal{A} \textit{ such that } &\dot{X}(t)=r(t)X(t)+\pi'(t)\sigma(t)\theta(t)\\
\textit{ and } &\Lambda_X(t)=\sigma'(t)\pi(t) \textit{ for }\forall t\in[0,T], \ \mathbb{P}-a.e. \rbrace.
\end{flalign*}
The set $\mathcal{U}(X)$ contains all  admissible controls $\pi\in {\cal A}$ that make $X$ an admissible wealth process. Note that $\mathcal{U}(X)\neq\varnothing$ if and only if $(\dot{X}(t), \Lambda_X(t)) \in {\cal S}(t, X(t))$  for $\left(\mathbb{P}\otimes Leb\right)$-a.e. $ (\omega,t)\in \Omega\times [0,T]$, where ${\cal S}$ is a set valued function defined by
$$ {\cal S}(\omega, t,x):=  \{(v,\xi): v=r( t)x+\xi'\theta( t) \textit{ and } \left[\sigma'\right]^{-1}( t)\xi\in K\}.$$
 Define the penalty function 
 $ L:\Omega\times [0,T]\times\mathbb{R}\times\mathbb{R}\times\mathbb{R}^N \rightarrow [0,\infty] $ by
 $$L(\omega, t,x,v,\xi) =
 f\left(\omega, t,x,[\sigma']^{-1}( t)\xi\right)
+\Psi_{{\cal S}(\omega, t,x)}( v,\xi)
$$
and 
the penalty function $ l_0:\mathbb{R}\rightarrow [0,\infty] $ by
\begin{equation*}
l_0(x) =\Psi_{\{x_0\}}(x),
\end{equation*}
where $\Psi_{U}( u)$ is a penalty function which equals 0 if $u$ is in set $U$ and $+\infty$ otherwise. 

For $ X  \in \mathbb{B}$, define the cost functional as
\begin{equation*}
\Phi(X) := l_0(x_0)+E\left[ g(X(T)) \right] + E\left[ \int_0^T L(t,X(t),\dot{X}(t),\Lambda_X(t))dt \right].
\end{equation*}
Note that $\Phi(X) =\infty$ if $X(0)\neq x_0$ or $\mathcal{U}(X)= \varnothing$. 
Problem \eqref{primal_ori} is equivalent to
\begin{equation*}
\mbox{Minimize $\Phi(X)$ subject to $X\in\mathbb{B}$. }
\end{equation*}
We now establish the dual problem over the set $ \mathbb{B} $. Define the following convex conjugate functions
\begin{flalign*}
m_0(y)&:=\sup_{x\in\mathbb{R}}\{ xy-l_0(x) \},\\
m_T(\omega, y)&:=\sup_{x\in\mathbb{R}}\lbrace -xy-g(\omega, x)\rbrace,\\
M(\omega, t,y,s,\gamma)&:=\sup_{x,v\in\mathbb{R},\xi\in\mathbb{R}^N}\lbrace xs+vy+\xi'\gamma-L(\omega, t,x,v,\xi) \rbrace,
\end{flalign*}
for all $  (\omega, t,y,s,\gamma) \in \Omega\times [0,T] \times\mathbb{R}\times\mathbb{R}\times\mathbb{R}^N $.
For each $ Y\equiv (y,\dot Y,\Lambda_Y)\in\mathbb{B} $, define
\begin{equation*}
\Psi(Y):= m_0(y)+E\left[ m_T(Y(T)) \right] + E\left[ \int_0^T M(t, Y(t), \dot Y(t), \Lambda_Y(t))dt \right].
\end{equation*}
Then the dual problem is given by
\begin{equation*}
\mbox{Minimize $\Psi(Y)$ subject to $Y\in\mathbb{B}$.}
\end{equation*}
We can write the dual problem equivalently as a stochastic control problem.
Some simple calculus gives 
\begin{eqnarray}
m_0(y)&=&x_0y,\nonumber\\
m_T(\omega, y)&=&\dfrac{(y+c)^2}{2a},\nonumber \\
M(\omega, t,y,s,\gamma) &=& \phi(t,s+r( t)y,\sigma( t)\left[\theta(t)y +\gamma\right]), \label{standard_M}
\end{eqnarray}
where $\phi$ is the conjugate function of $\tilde f(\omega, t,x,\pi)=f(\omega, t,x,\pi)+\Psi_K(\pi)$, namely,
$$ \phi(\omega, t,\alpha,\beta):=\sup_{x\in \mathbb{R}, \pi\in K}\{ x\alpha + \pi' \beta - f(\omega, t,x,\pi)\}.$$
The dual control problem is therefore given by
\begin{equation}\label{dual_cost_function}
\mbox{Minimize }  \tilde{\Psi}(y,\alpha,\beta):=  m_0(y)+E\left[ m_T(Y(T)) \right] + E\left[ \int_0^T \phi(t,\alpha(t),\beta(t))dt \right],
\end{equation}
where $ Y$ satisfies 
\begin{equation}\label{dual_fsde}
\left\{
\begin{array}{l}
dY(t)=\left[\alpha(t)-r(t)Y(t)\right]dt+\left[ \sigma^{-1}(t)\beta(t)-\theta(t)Y(t) \right]' dW(t)\vspace{2mm}\\
Y(0)=y.
\end{array}
\right.
\end{equation}
Here we have used the relation (\ref{standard_M}) to get $\alpha(t)=\dot Y(t)+r(t)Y(t)$ and $\beta(t)=\sigma(t)(\theta(t)Y(t)+\Lambda_Y(t))$, which shows $\dot Y(t)=\alpha(t)-r(t)Y(t)$ and $\Lambda_Y(t) = \sigma^{-1}(t)\beta(t) -\theta(t)Y(t)$,  for the dual process $Y$.
The dual control process for $Y$ is $(y,\alpha,\beta)\in\mathbb{B} $. From \cite[Corollary 2.5.10]{kylovcontrolled}, we have 
$Y^{(y,\alpha,\beta)}\in {\cal S}^2(0,T; \mathbb{R})$.
Note that the control constraint is implicit for the dual problem.  For example, if $Q=0,S=0,R=0$, then $\alpha$ must be zero and may be simply dropped in  (\ref{dual_cost_function}) and (\ref{dual_fsde}).

\begin{remark} (Alternative way of deriving the dual problem) We have followed \cite{heunislabbe:constrainedMV} to derive the dual problem (\ref{dual_cost_function}) and (\ref{dual_fsde}) by first converting the primal dynamic problem into a static problem, then applying convex analysis to get the static dual problem, and finally recovering the dual dynamic problem. One may derive the dual problem (\ref{dual_cost_function}) and (\ref{dual_fsde}) directly using a standard method in utility maximization in mathematical finance. Specifically, we may assume the dual process $Y$ is driven by a SDE
$$ dY(t)=\alpha_1(t)dt+ \beta_1(t)dW(t)$$
with initial condition $Y(0)=y$, where $\alpha_1$ and $\beta_1$ are two stochastic processes to be determined. Ito's lemma gives
$$ d(X^\pi(t)Y(t))= 
(X^\pi(t)\alpha(t)+\pi'(t)\beta(t))dt + \mbox{local martingale},
$$
where $\alpha(t)=\alpha_1(t)+r(t)Y(t)$ and $\beta(t)=\sigma(t)(\beta_1(t)+\theta(t)Y(t))$. Since $\alpha_1(t)=\alpha(t)-r(t)Y(t)$ and $\beta_1(t)=\sigma^{-1}(t)\beta(t)-\theta(t)Y(t)$, we have $Y$ satisfies SDE (\ref{dual_fsde}). The process
$X^\pi(t)Y(t)-\int_0^t (X^\pi(s)\alpha(s)+\pi'(s)\beta(s))ds$ is a local martingale and a super-martingale if we assume further that it is bounded below by an integrable process, in particular, we have the relation
\begin{equation} \label{eqn2.9}
E\left[X^\pi(T)Y(T)-\int_0^T (X^\pi(s)\alpha(s)+\pi'(s)\beta(s))ds\right]\leq X^\pi(0)Y(0)=x_0y.
\end{equation}
The constrained minimization problem (\ref{primal_ori}) can be written equivalently  as
$$ \max_\pi E\left[\int_0^T (-f(t,X^\pi(t),\pi(t)) -\Psi_K(\pi(t))) dt  - g(X^\pi(T))\right].$$ 
 The dual functions of $-f(t,\cdot,\cdot)-\Psi_K(\cdot)$ and $-g(\cdot)$ are given by
$$ \phi(t,\alpha,\beta)=\sup_{x,\pi}\{ -f(t,x,\pi)-\Psi_K(\pi)+x\alpha+\pi'\beta\}
\mbox{ and } m_T(y)=\sup_{x}(-g(x)-xy).$$
Combining the dual relations above and (\ref{eqn2.9}), we have
\begin{eqnarray*}
&& \max_\pi E\left[\int_0^T (-f(t,X^\pi(t),\pi(t)) -\Psi_K(\pi(t))) dt  - g(X^\pi(T))\right]\\
&\leq& \min_{y,\alpha,\beta} \left\{x_0y + E\left[\int_0^T \phi(t,\alpha(t),\beta(t))dt + m_T(Y(T))\right]\right\},
\end{eqnarray*}
which gives the dual problem (\ref{dual_cost_function}).
\end{remark}

\begin{remark} (Existence of optimal solutions)\label{rk-solu}
Following a similar argument as in \cite[Proposition 5.4]{heunislabbe:constrainedMV}, we can show that $ \tilde{\Psi} $ defined in  \eqref{dual_cost_function} is convex on $ \mathbb{B} $
due to convexity of $m_T$ and $\phi$ and  linearity of state equation \eqref{dual_fsde}. Since $ \tilde{\Psi}(y,\alpha,\beta)>x_0y>-\infty$ for all $(y,\alpha,\beta)\in\mathbb{B} $ and $ \tilde{\Psi}(0,0,0) = E\left[\frac{c^2}{2a}\right]<\infty $, we  have $ \tilde{\Psi} $ is proper. Furthermore, by the non-negativity and semi-continuity of $ \phi $ and Fatou's lemma, we conclude that $ \tilde{\Psi} $ is lower semi-continuous on $ \mathbb{B} $. Finally, using It\^{o}\rq{}s isometry, we can show that $ \tilde{\Psi} $ is coercive (i.e., $ \tilde{\Psi}(y,\alpha,\beta) \rightarrow \infty \textit{ as }\| (y,\alpha,\beta)\| \rightarrow \infty$). Hence, the existence of a solution to the dual problem is guaranteed from \cite[Chapter 2, Proposition 1.2]{ekelandtemam:convexanalysis}, that is, there exists some $(\hat{y},\hat{\alpha},\hat{\beta})\in\mathbb{B}$ such that 
\begin{equation*}
\inf_{(y,\alpha,\beta)\in\mathbb{B}} \tilde{\Psi}(y,\alpha,\beta) = \tilde{\Psi}(\hat{y},\hat{\alpha},\hat{\beta}) \in \mathbb{R}.
\end{equation*}
Given the dual optimal solution $(\hat{y},\hat{\alpha},\hat{\beta})$,  we may apply Theorem \ref{dual_to_primal} to construct an optimal control  $\hat \pi$ for problem (\ref{primal_ori}), which proves the existence of a solution to the primal problem. It is more difficult, but accomplishable, to prove directly the existence of a solution to the primal problem as one needs to show the closedness of the set of all replicable terminal wealth under  pointwise control constraints, see \cite{Czichowsky:convexdualconstraints} (also \cite{cwz}) for detailed discussions. 
\end{remark}

\section{Main Results} \label{mainresults}
In this section, we derive the necessary and sufficient optimality conditions for  primal and dual problems and  show the connection between the optimal solutions  through their corresponding FBSDEs. To highlight the main results and streamline the discussion, we leave the proofs of all the theorems  in Section~5.

Given any admissible control $\pi\in {\cal A}$ and solution $X^\pi$  to the SDE (\ref{wealth_process}), 
 the associated adjoint equation in  unknown  processes $ p_1\in {\cal S}^2(0,T; \mathbb{R})$ and $q_1\in {\cal H}^2(0,t;\mathbb{R}^N)$  is the following linear BSDE 
\begin{equation}\label{adjoint_BSDE}
\left\{
\begin{array}{l}
dp_1(t)=\left[-r(t)p_1(t)+Q(t)X(t)+S'(t)\pi(t) \right]dt + q_1'(t)dW(t)\vspace{2mm}\\
p_1(T)=-aX^{\pi}(T)-c.
\end{array}
\right.
\end{equation}
From \cite[Theorem 6.2.1]{pham:contimuoustimeSC} we know that 
there exists a unique solution $ (p_1,q_1) $ to the  BSDE \eqref{adjoint_BSDE}.  
We now state the necessary and sufficient conditions for the primal problem.
\begin{theorem}(Primal problem and associated FBSDE)\label{primal_necesuff}
Let $ \hat{\pi}\in\mathcal{A} $. Then $ \hat{\pi} $ is optimal for the primal problem if and only if the solution $ (X^{\hat{\pi}},\hat{p}_1,\hat{q}_1) $ of FBSDE
\begin{equation}\label{primal_necesuff_FBSDE}
\left\{
\begin{array}{l}
dX^{\hat{\pi}}(t)=\left[ r(t)X^{\hat{\pi}}(t)+\hat{\pi}'(t)\sigma(t)\theta(t) \right] dt +\hat{\pi}'(t)\sigma(t)dW(t)\\
X^{\hat{\pi}}(0)=x_0\\
d\hat{p}_1(t) = \left[-r(t)\hat{p}_1(t)+Q(t)X^{\hat{\pi}}(t)+S'(t)\hat{\pi}(t)\right]dt + \hat{q}_1'(t)dW(t)\\
\hat{p}_1(T) = -aX^{\hat{\pi}}(T)-c
\end{array}
\right.
\end{equation} 
satisfies the condition
\begin{equation}\label{primal_condition}
\left[ \hat{\pi}'-\pi' \right]\left[ \hat{p}_1(t)\sigma(t)\theta(t)+\sigma(t)\hat{q}_1(t)+S(t)X^{\hat{\pi}}(t)+R(t)\hat{\pi}(t) \right]\geq 0
\end{equation}
for $\left(\mathbb{P}\otimes Leb\right)$-a.e. $ (\omega,t)\in \Omega\times [0,T] $ and $\pi\in K$.
\end{theorem}

\begin{remark} \label{rk5}
If one knows the optimal control $\hat\pi$, it is easy to find the optimal wealth process $X^{\hat\pi}$ and the adjoint process $(\hat p_1, \hat q_1)$ as  (\ref{primal_necesuff_FBSDE}) is a decoupled linear FBSDE given $\hat\pi$. It is much more difficult, but most interesting, to find the optimal control $\hat \pi$ using (\ref{primal_necesuff_FBSDE}) and (\ref{primal_condition}), which is related to the solvability of a fully coupled  constrained linear FBSDE. From Remark
\ref{rk-solu}, we know there exists a solution $(X^{\hat\pi}, \hat p_1, \hat q_1)$ to the constrained FBSDE 
(\ref{primal_necesuff_FBSDE}) and (\ref{primal_condition}). It is a  challenge on how one may actually find the solution.
If  $K=\mathbb{R}^N$, then condition (\ref{primal_condition}) becomes 
$$ \hat{p}_1(t)\sigma(t)\theta(t)+\sigma(t)\hat{q}_1(t)+S(t)X^{\hat{\pi}}(t)+R(t)\hat{\pi}(t) =0.$$
If we further assume $R(t)$ is positive definite and $R(t)^{-1}$ is uniformly bounded, then we can substitute the optimal control $\hat \pi(t)$ into the FBSDE (\ref{primal_necesuff_FBSDE}) to get a fully-coupled linear FBSDE with random coefficients, see \cite{yong2006} for discussions on the solvability of linear FBSDEs.
\end{remark}

Given any admissible control $ (y,\alpha,\beta)\in\mathbb{B} $ and solution $ Y^{(y,\alpha,\beta)} $ to the SDE \eqref{dual_fsde},  the associated adjoint equation in  unknown  processes $ p_2\in {\cal S}^2(0,T; \mathbb{R})$ and $q_2\in {\cal H}^2(0,t;\mathbb{R}^N)$  is the following linear BSDE 
\begin{equation}\label{dual_adjoint_BSDE}
\left\{
\begin{array}{l}
dp_2(t)=\left[ r(t)p_2(t)+q_2'(t)\theta(t) \right]dt + q_2'(t)dW(t)\\
p_2(T) = -\dfrac{Y^{(y,\alpha,\beta)}(T)+c}{a}.
\end{array}
\right.
\end{equation}
From \cite[Theorem 6.2.1]{pham:contimuoustimeSC}, we know that there exists a unique solution $ (p_2,q_2) $ to the BSDE
 \eqref{dual_adjoint_BSDE}. To derive the necessary condition, 
we need to impose the following assumption on $ \phi $ at the optimal dual control process
$(\hat\alpha, \hat\beta)$.

  \begin{assumption}\label{assum_monoton} 
  Let $(\hat\alpha, \hat\beta)$ be given and $ \alpha,\beta $ be any admissible control. Then there exists a   $ Z\in {\cal P}(0,T; \mathbb{R})$ satisfying $ E[\int_0^T\vert Z(t)\vert dt] <\infty$ and 
 \begin{equation}\label{assumptionZ}
 Z(t)\geq\dfrac{\phi(t,\hat{\alpha}(t)+\varepsilon\alpha(t),\hat{\beta}(t)+\varepsilon\beta(t))-\phi(t,\hat{\alpha}(t),\hat{\beta}(t))}{\varepsilon}
 \end{equation} 
 for $ \left(\mathbb{P}\otimes Leb\right) $-a.e. $ (\omega,t)\in \Omega\times [0,T] $ and $  \varepsilon\in(0,1] $.
 \end{assumption}

\begin{remark} \label{remark4} Here are a few comments on Assumption \ref{assum_monoton}.
\begin{enumerate}
\item Condition (\ref{assumptionZ}) is a technical condition that ensures one can apply the monotone convergence theorem and pass the limit under the expectation and integral as $\varepsilon\downarrow 0$, which is used in proving  the second and third relations in (\ref{dual_conditions}), see the proof of Theorem  \ref{Dual_Nece_Cond} in Section \ref{section5}. 
   A similar assumption is used in \cite[Assumption 1.2]{cadenillaskaratzas:convexsmp} on the data of the primal problem.   
\item If $K=\mathbb{R}^N$, $S(t)=0$ and  $Q(t),R(t)$ are positive definite  and their inverses are uniformly bounded, then $\phi(t, \alpha, \beta)={1\over 2}Q(t)^{-1}\alpha^2 + {1\over 2}\beta' R(t)^{-1}\beta$.  Condition (\ref{assumptionZ})   holds if $Z$ is chosen to be
$$ Z(t):=Q(t)^{-1}\hat\alpha(t)\alpha(t) + \hat\beta'(t)R(t)^{-1}\beta(t)
+{1\over 2} Q(t)^{-1}\alpha(t)^2 + {1\over 2}\beta'(t)R(t)^{-1}\beta(t).$$

\item If $Q(t)=0,S(t)=0,R(t)=0$, then $\alpha(t)=0$ for the dual problem. We may drop $\alpha$ in the expression of $\phi$ which becomes a support function of $K$, i.e., $\phi$ is given by $\phi(t,\beta)=\delta(\beta):=\sup_{\pi\in K} \pi'\beta$. If we further assume that $K$ is a bounded set, then condition (\ref{assumptionZ})   holds if  $Z$ is chosen to be $Z(t)=\delta(\beta(t))$. However, if $K$ is unbounded, then Assumption \ref{assum_monoton} may not hold and we cannot use the monotone convergence theorem to prove (\ref{dual_conditions}). Other methods may have to be used, see Remark \ref{remark10} for further discussions. 
\end{enumerate}
\end{remark}

We now  state the necessary and sufficient conditions for the dual  problem.

\begin{theorem}(Dual problem and associated FBSDE)\label{Dual_Nece_Cond}
Let $ (\hat{y},\hat{\alpha},\hat{\beta}) \in\mathbb{B}$ satisfy Assumption \ref{assum_monoton}. Then $ (\hat{y},\hat{\alpha},\hat{\beta})$ is  optimal for the dual problem if and only if  the solution $ (Y^{(\hat{y},\hat{\alpha},\hat{\beta})},\hat{p}_2,\hat{q}_2) $ of FBSDE
\begin{equation}\label{dual_FBSDE}
\left\{
\begin{array}{l}
dY^{(\hat{y},\hat{\alpha},\hat{\beta})}(t)=\left[\hat{\alpha}(t)-r(t)Y^{(\hat{y},\hat{\alpha},\hat{\beta})}(t)\right]dt + \left[ \sigma^{-1}(t)\hat{\beta}(t)-\theta(t)Y^{(\hat{y},\hat{\alpha},\hat{\beta})}(t) \right]' dW(t)\vspace{2mm}\\
Y^{(\hat{y},\hat{\alpha},\hat{\beta})}(0)=\hat{y}\vspace{2mm}\\
d\hat{p}_2(t)=\left[ r(t)\hat{p}_2(t)+\hat{q}_2'(t)\theta(t) \right] dt + \hat{q}_2'(t)dW(t)\vspace{2mm}\\
\hat{p}_2(T) = -\dfrac{Y^{(\hat{y},\hat{\alpha},\hat{\beta})}(T)+c}{a}
\end{array}
\right.
\end{equation}
 satisfies the conditions
\begin{equation}\label{dual_conditions}
\left\{
\begin{array}{l}
\hat{p}_2(0)=x_0,\vspace{2mm}\\
\left[\sigma'\right]^{-1}(t)\hat{q}_2(t)\in K,\vspace{2mm}\\
\left( \hat{p}_2(t),\left[\sigma'\right]^{-1}(t)\hat{q}_2(t) \right)\in \partial\phi(\hat{\alpha}(t),\hat{\beta}(t)),
\end{array}
\right.
\end{equation}
for $ \left(\mathbb{P}\otimes Leb\right) $-a.e. $ (\omega,t)\in \Omega\times [0,T] $.
\end{theorem}

\begin{remark} 
Similar discussions as in Remark \ref{rk5} apply here. If one knows the optimal control $(\hat y,\hat\alpha, \hat\beta)$, it is easy to find the optimal dual process $Y$ and the adjoint process $(\hat p_2, \hat q_2)$ as (\ref{dual_FBSDE}) is a decoupled linear FBSDE given $(\hat y,\hat\alpha, \hat\beta)$. It is much more difficult, but most interesting, to find the optimal control $(\hat y,\hat\alpha, \hat\beta)$ using (\ref{dual_FBSDE}) and (\ref{dual_conditions}).
If $K=\mathbb{R}^N$, $S(t)=0$ and  $Q(t),R(t)$ are positive definite, then from Remark \ref{remark4}, $\phi$ is a quadratic function of $\alpha$ and $\beta$ and we can write optimal controls $\hat\alpha$ and $\hat\beta$ in terms of adjoint processes $\hat p_2$ and $\hat q_2$. The FBSDE (\ref{dual_FBSDE}) becomes a fully coupled linear FBSDE with an additional condition $\hat p_2(0)=x_0$, which is used to determine the constant control $\hat y$. 
\end{remark}


We can now state the dynamic relations of the optimal portfolio and wealth processes of the primal problem and the adjoint processes of the dual problem and vice versa.

\begin{theorem}\label{dual_to_primal}(From dual problem to primal problem)
Suppose that $ (\hat{y},\hat{\alpha},\hat{\beta}) $ is optimal for the dual problem. Let $ \left( Y^{(\hat{y},\hat{\alpha},\hat{\beta})}, \hat{p}_2,\hat{q}_2\right)  $  be the associated process that satisfies the FBSDE \eqref{dual_FBSDE} and condition \eqref{dual_conditions}. Define 
\begin{equation}\label{relations_pi_defn}
\hat{\pi}(t):= \left[\sigma'\right]^{-1}(t)\hat{q}_2(t),\ t\in[0,T].
\end{equation} 
Then $ \hat{\pi} $ is the optimal control for the primal problem with initial wealth $ x_0 $. The optimal wealth process and associated adjoint processes are given by
\begin{equation}\label{dual_to_primal_soln}
\left\{
\begin{array}{l}
X^{\hat{\pi}}(t)=\hat{p}_2(t),\\
\hat{p}_1(t)=Y^{(\hat{y},\hat{\alpha},\hat{\beta})}(t),\\
\hat{q}_1(t)=\sigma^{-1}(t)\hat{\beta}(t)-\theta(t)Y^{(\hat{y},\hat{\alpha},\hat{\beta})}(t) \textit{  for }\forall t\in[0,T].
\end{array}
\right.
\end{equation}
\end{theorem}

\begin{remark} 
The key benefit of Theorem \ref{dual_to_primal} is that if one can solve the dual problem, then one can get automatically  the optimal control $\hat\pi$ for the primal problem using (\ref{relations_pi_defn}). As discussed in Remark \ref{rk5}, it is in general  difficult to find the optimal control $\hat\pi$ using (\ref{primal_necesuff_FBSDE}) and (\ref{primal_condition}) directly. Section \ref{sec4.2} provides an example to illustrate this point.
\end{remark}

\begin{theorem}\label{primaltodual}(From primal problem to dual problem)
Suppose that $ \hat\pi\in\mathcal{A} $ is optimal for the primal problem with initial wealth $ x_0 $. Let $\left( X^{\hat{\pi}},\hat{p}_1,\hat{q}_1\right)$ be the associated process that satisfies the FBSDE \eqref{primal_necesuff_FBSDE} and  condition \eqref{primal_condition}. Define
\begin{equation}
\left\{
\begin{array}{l}
\hat{y}=\hat{p}_1(0),\vspace{2mm}\\
\hat{\alpha}(t)=Q(t)X^{\hat{\pi}}(t)+S'(t)\hat{\pi}(t),\vspace{2mm}\\
\hat{\beta}(t)=\sigma(t)\left[ \hat{q}_1(t)+\theta(t)\hat{p}_1(t) \right].
\end{array}\label{primal_to_dual_control}
\right.
\end{equation}
Then $ (\hat{y},\hat{\alpha},\hat{\beta}) $ is the optimal control for the dual problem.
The optimal dual state process and associated adjoint processes are given by
\begin{equation}\label{primal_to_dual_cond}
\left\{
\begin{array}{l}\vspace{1mm}
Y^{(\hat{y},\hat{\alpha},\hat{\beta})}(t)=\hat{p}_1(t),\\\vspace{1mm}
\hat{p}_2(t)=X^{\hat{\pi}}(t),\\\vspace{1mm}
\hat{q}_2(t)=\sigma'(t)\hat{\pi}(t).
\end{array}
\right.
\end{equation}
\end{theorem}

\begin{remark}
The main results (Theorems \ref{primal_necesuff}, \ref{Dual_Nece_Cond}, \ref{dual_to_primal} and \ref{primaltodual}) can be extended to utility maximization problems with quadratic cost functions being replaced by utility functions. The ideas are similar but proofs are much more complicated as utility functions are only defined on the positive real line with unbounded non-Lipschitz derivatives, in contrast to quadratic functions which are defined on the whole real line with linear derivatives. The authors have a separate paper discussing the details of dynamic relations of the primal and dual problems for general utility functions with control constraints via maximum principles and FBSDEs and will publish the results elsewhere. 
\end{remark}

\section{Quadratic Risk Minimization with Cone Constraints}
In this section we consider the following quadratic risk minimization problem: 
\begin{equation}\label{quad_min_uncon_primal}
\left\{
\begin{array}{l}
\textit{Minimize } J(\mathbb{\pi}(\cdot))=E\left[ \dfrac{1}{2} aX(T)^2 \right],\vspace{1.5mm} \\
\textit{Subject to } \left(X(\cdot),\pi(\cdot)\right) \textit{ is admissible}.
\end{array}
\right.
\end{equation}
Assume $ K\subset \mathbb{R}^N $ is a closed convex cone.
The dual problem is given by
\begin{equation}\label{quad_min_uncon_dual}
\textit{Minimize }x_0y+E\left[  \dfrac{Y(T)^2}{2a}  \right]+E\left[ \int_0^T \delta(\beta(t))dt \right]
\end{equation}
over $ (y,\beta)\in\mathbb{R}\times {\cal H}^2(0,T; \mathbb{R}^N)$, where  $ Y $ satisfies the SDE \eqref{dual_fsde} with $\alpha(t)=0$ and $\delta(\beta)=\sup_{\pi\in K} \pi'\beta$,  the support function of  $K$. \cite[Proposition 5.4]{heunislabbe:constrainedMV} states that there exists an optimal control $ (\hat{y},\hat{\beta}) $ to \eqref{quad_min_uncon_dual} with associated optimal state process $ \hat Y$. 

\begin{remark}\label{remark10}
Since  $K$ is unbounded,  Assumption \ref{assum_monoton} may not hold.  Using the subadditivity and positive homogeneity of $\delta$, we have  (see (\ref{necedual_proof_5}))
\begin{equation}\label{temp1}
E\left[ \int_0^T \left[ \delta(\beta(t))-\hat{q}_2'(t)\sigma^{-1}(t){\beta}(t) \right]dt \right] \geq 0. 
\end{equation}
Let $B:=\{(\omega,t)\in \Omega\times [0,T]: \left[\sigma'\right]^{-1}(t)\hat{q}_2(t)\in K\}$. By \cite[Lemma 5.4.2]{karatzasshreve:mathfinance}, there exists $\nu\in {\cal P}(0,T; \mathbb{R}^N)$ such that $|\nu(t)|\leq 1$ and $|\delta(\nu(t))|\leq 1$ and 
\begin{eqnarray*}
 \left[\sigma'\right]^{-1}(t)\hat{q}_2(t)\in K &\Leftrightarrow& \nu(t)=0,\\
\left[\sigma'\right]^{-1}(t)\hat{q}_2(t)\not\in K  &\Leftrightarrow&
\delta(\nu(t))-\hat{q}_2'(t)\sigma^{-1}(t){\nu}(t)<0
\end{eqnarray*}
 for $ \left(\mathbb{P}\otimes Leb\right) $-a.e. $ (\omega,t)\in \Omega\times [0,T] $.
The existence of $\nu$ ensures that the complement set of $B$ has measure zero on $\Omega\times [0,T]$ (otherwise there is a contradiction to (\ref{temp1})). Hence we conclude 
$ \left[\sigma'\right]^{-1}(t)\hat{q}_2(t)\in K$ for $ \left(\mathbb{P}\otimes Leb\right) $-a.e. The third relation in (\ref{dual_conditions}) can also be proved directly. 
\end{remark}
\subsection{Random coefficient case}\label{CQRM}
We  have the following result. 
\begin{lemma}
Let $ (\hat{y},\hat{\beta}) $ be the optimal control of the dual problem \eqref{quad_min_uncon_dual} and $ \hat{Y} $ be the corresponding optimal state process. Then $ \hat{\beta}(t)=0 $ if $ \hat{Y}(t)=0 $ for $ \left(\mathbb{P}\otimes Leb\right) $-a.e. $ (\omega,t)\in \Omega\times [0,T] $.
\end{lemma}
\begin{proof}
Applying Ito's formula to $ \hat Y(t)^2 $, we get
\begin{align*}
d\hat Y(t)^2=&\left[-2r(t)\hat Y(t)^2+\left( \sigma^{-1}(t)\hat{\beta}(t)-\theta(t)\hat Y(t) \right)'\left( \sigma^{-1}(t)\hat{\beta}(t)-\theta(t)\hat Y(t) \right)\right]dt\\
&{} +2\hat Y(t)\left[ \sigma^{-1}(t)\hat{\beta}(t)-\theta(t)\hat Y(t) \right]' dW(t).
\end{align*}
Define the process
\begin{equation*}
\tilde S(t):= \int_0^t 2\hat Y(s)[ \sigma^{-1}(s)\hat{\beta}(s)-\theta(s)\hat Y(s) ]' dW(s).
\end{equation*}
Following a similar argument as in the proof of Theorem \ref{primal_necesuff}, we know $\tilde  S $ is a martingale. Taking expectation of $ \dfrac{\hat Y(T)^2}{2a} $, we have
\begin{align*}
&E\left[\dfrac{\hat Y(T)^2}{2a}\right] := \\
&E \left[\dfrac{\hat{y}^2}{2a}\right]
+E\left[\int_0^T\bigg[-\dfrac{r(t)\hat Y(t)^2}{a}+\dfrac{\left( \sigma^{-1}(t)\hat{\beta}(t)-\theta(t)\hat Y(t) \right)'\left( \sigma^{-1}(t)\hat{\beta}(t)-\theta(t)\hat Y(t) \right)}{2a}\bigg]dt\right].
\end{align*}
Define the set 
\begin{equation*}
\Pi:= \left\{ (\omega,t)\in  \Omega\times [0,T] : \hat Y(t)=0, \hat{\beta}(t)\neq 0 \right\}.
\end{equation*}
We must have  $ \left(\mathbb{P}\otimes Leb\right)(\Pi)=0 $,  otherwise, 
we may replace $\hat{\beta}(t)$ by 0 on the set $\Pi$ and keep the same   $\hat{\beta}(t)$ on the complement of $\Pi$, then we get the dual value strictly less than the one using   $\hat{\beta}(t)$ everywhere, which is a  contradiction to the optimality of $\hat{\beta}(t)$.
\end{proof}
Let $ \hat{\beta}(t)=\hat{\gamma}(t)\hat Y(t) $ for $  t\in[0,T] $. Then $ \hat Y$ follows the SDE
\begin{equation*}
\left\{
\begin{array}{l}
d\hat Y(t)=-r(t)\hat Y(t)dt+\left[\sigma^{-1}(t)\hat{\gamma}(t)-\theta(t)\right]' \hat Y(t)dW(t)\\
\hat Y(0)=\hat{y}.
\end{array}
\right.
\end{equation*}
Hence, we have $ \hat Y(t)=\hat{y}\hat{H}(t) $, where
\begin{align*}
\hat{H}(t):= \exp\bigg(& \int_0^t \left[ -r(s)-\dfrac{1}{2}\left(\sigma^{-1}(s)\hat{\gamma}(s)-\theta(s)\right)'\left(\sigma^{-1}(s)\hat{\gamma}(s)-\theta(s)\right) \right]ds \\
& + \left[\sigma^{-1}(s)\hat{\gamma}(s)-\theta(s)\right]' dW(s) \bigg).
\end{align*}
Let $ \Gamma $ satisfy the linear SDE 
\begin{equation*}
d\Gamma(t)=\Gamma(t)[-r(t)dt-\theta'(t)dW(t)], \ \Gamma(0)=1.
\end{equation*}
By Theorem \ref{dual_to_primal}, also noting $\Gamma(t)\hat p_2(t)$ is a martingale, we obtain
\begin{equation*}
\hat{p}_2(0)=E\left[ \Gamma(T)\hat p_2(T) \right]= E\left[ -\Gamma(T)\dfrac{\hat Y(T)}{a} \right]=-\hat{y}E\left[\Gamma(T)\dfrac{\hat{H}(T)}{a}\right]=x_0,
\end{equation*}
which implies 
\begin{equation*}
\hat{y}=-\dfrac{x_0}{E\left[\dfrac{\Gamma(T)\hat{H}(T)}{a}\right]}.
\end{equation*}
Moreover, we have
$$
\hat{p}_2(t)= \Gamma(t)^{-1}E\left[ -\Gamma(T)\dfrac{\hat Y^(T)}{a} \bigg| \mathcal{F}_t \right] 
=-\hat{y}\Gamma(t)^{-1}E\left[ \Gamma(T)\dfrac{\hat H^(T)}{a} \bigg| \mathcal{F}_t \right],  
$$
which shows that 
 $ \hat{p}_2(t)\neq 0 \ \mathbb{P}$-a.e. for $ t\in[0,T] $.

Suppose $ x_0>0 $, then $ \hat{Y}(t)<0 $ and $ \hat{p}_2(t)>0 $ for $ \forall t\in[0,T], \ \mathbb{P} $-a.e. Define
\begin{equation*}
P_+(t):= -\dfrac{\hat Y(t)}{\hat{p}_2(t)}=-\dfrac{\hat{p}_1(t)}{\hat{X}(t)}, \ \forall t\in[0,T].
\end{equation*}
Applying Ito's formula, we have
\begin{align}\nonumber
dP_+(t)=&\left[ -2r(t)P_+(t)-P_+(t)\dfrac{\hat{\pi}'(t)}{\hat{X}(t)}\sigma(t)\theta(t)+\dfrac{\pi'(t)\sigma(t)\hat{q}_1(t)}{\hat{X}(t)^2} + \dfrac{P_+(t)\pi'(t)\sigma(t)\sigma'(t)\pi(t)}{\hat{X}(t)^2} \right]dt\\\nonumber
&+\left[ -\dfrac{\hat{q}_1(t)}{\hat{X}(t)}-P_+(t)\sigma'(t)\dfrac{\pi(t)}{\hat{X}(t)} \right]' dW(t),\\\label{quad_min_con_eqn1}
=&\left[ -2r(t)P_+(t)-\hat{\xi}_+'(t)\left( \sigma(t)\theta(t)P_+(t)+\sigma(t)\Lambda_+(t) \right) \right]dt + \Lambda_+'(t)dW(t),
\end{align}
where 
\begin{equation*}
\Lambda_+(t):= -\dfrac{\hat{q}_1(t)}{\hat{X}(t)}-\dfrac{P_+(t)\sigma'(t)\pi(t)}{\hat{X}(t)}, \ \hat{\xi}_+(t):= \dfrac{\hat{\pi}(t)}{\hat{X}(t)}.
\end{equation*}
Define 
\begin{align*}
H_+(t,v,P,\Lambda):=& v' P\sigma(t)\sigma'(t)v + 2v'\left[ \sigma(t)\theta(t)P+\sigma(t)\Lambda \right],\\
H^*_+(t,P,\Lambda):=& \inf_{v\in K}H_+(t,v,P,\Lambda).
\end{align*}
We have
\begin{align*}
\partial_v H_+(t,\hat{\xi}_+(t),P_+(t),\Lambda_+(t))= &2\left[ P_+(t)\sigma(t)\sigma'(t)\dfrac{\hat{\pi}(t)}{\hat{X}(t)}+\sigma(t)\theta(t)P_+(t)+\sigma(t)\Lambda_+(t) \right]\\
=& 2\left[ -\sigma(t)\dfrac{\hat{q}_1(t)}{\hat{X}(t)}-\sigma(t)\theta(t)\dfrac{\hat{p}_1(t)}{\hat{X}(t)} \right].
\end{align*}
Recall that by Theorem \ref{primal_necesuff}, we have
\begin{equation}\label{quad_min_con_dual_BSDE_eqn1}
[\hat{\pi}(t)-\pi]'[\hat{p}_1(t)\sigma(t)\theta(t)+\sigma(t)\hat{q}_1(t)]\geq 0
\end{equation} 
for $\left(\mathbb{P}\otimes Leb\right) $-a.e. $ (\omega,t)\in \Omega\times [0,T] $ and $ \pi\in K $.
According to Theorem \ref{primaltodual}, $ \hat{X}(t)=\hat{p}_2(t)>0 $. Dividing both sides of \eqref{quad_min_con_dual_BSDE_eqn1}  by $ \hat{X}(t)^2 $, we obtain that 
\begin{equation*}
[\hat{\xi}_+(t)-\xi]'\partial_v H_+(t,\hat\xi_+(t),P_+(t),\Lambda_+(t))\leq 0
\end{equation*}
for $\left(\mathbb{P}\otimes Leb\right)$-a.e. $ (\omega,t)\in \Omega\times [0,T] $ and $ \xi\in K $. By \cite[Proposition 2.2.1]{ekelandtemam:convexanalysis}, we conclude that 
\begin{equation}\label{H_star_tosub}
H^*_+(t,P_+(t),\Lambda_+(t))=H_+(t,\hat{\xi}_+(t),P_+(t),\Lambda_+(t)) \ \forall t\in[0,T],\ \mathbb{P}-a.e.
\end{equation}
Moreover, by \cite[Page 52, Corollary]{clarke:optimization}, we have
\begin{equation*}
0\in P_+(t)\sigma(t)\sigma'(t)\hat{\xi}_+(t)+\sigma(t)[\theta(t)P_+(t)+\Lambda_+(t)]+N_K(\hat{\xi}_+(t)), \ \forall t\in[0,T]\ \mathbb{P}\textit{-a.e.}
\end{equation*}
where  $ N_K(x) :=\{ p\in\mathbb{R}^N: p'(x^*-x)\leq 0, \forall x^*\in K \} $, the normal cone of $K$ at $x\in K$. For all $p\in N_K(x) $, since $K$ is a cone, by choosing $x^*=2x$ and $x^*={1\over 2}x$, 
 we have $ p'x\leq 0 $ and $ -\frac{1}{2} p'x\leq 0 $, which gives $ p'x=0 $. Therefore  
\begin{equation}\label{normalconeequality}
\hat{\xi}_+'(t)P_+(t)\sigma(t)\sigma'(t)\hat{\xi}_+(t)+\hat{\xi}_+'(t)\sigma(t)[\theta(t)P_+(t)+\Lambda_+(t)]=0.
\end{equation}
Substituting \eqref{normalconeequality} into \eqref{H_star_tosub}, we obtain
\begin{equation}\label{quad_min_con_eqn2}
H^*_+(t,P_+(t),\Lambda_+(t))=\hat{\xi}_+'(t)\left[ \sigma(t)\theta(t)P_+(t)+\sigma(t)\Lambda_+(t)\right] \ \forall t\in[0,T].
\end{equation}
Substituting \eqref{quad_min_con_eqn2} back into \eqref{quad_min_con_eqn1}, we have that $ P_+ $ is the solution to the following nonlinear BSDE
\begin{equation}\label{quad_min_con_dual_BSDE_P+}
\left\{
\begin{array}{l}
dP_+(t)=-\left[ 2r(t)P_+(t)+H^*_+(t,P_+(t),\Lambda_+(t)) \right]dt+\Lambda'_+(t)dW(t),\\
P_+(T)=a,\\
P_+(t)>0. \ \forall t\in[0,T].
\end{array}
\right.
\end{equation}
Similarly, if $ x_0<0 $, then $ \hat Y(t)>0 $ and $ \hat{p}_2(t)<0 $ for $ t\in[0,T], \ \mathbb{P} $-a.e. Define
\begin{equation*}
P_-(t):= -\dfrac{\hat Y(t)}{\hat{p}_2(t)}=-\dfrac{\hat{p}_1(t)}{\hat{X}(t)}, \ \forall t\in[0,T].
\end{equation*} 
Using a similar approach, it can be shown that $ P_- $ is the solution to the following nonlinear BSDE
\begin{equation}\label{quad_min_con_dual_BSDE_P-}
\left\{
\begin{array}{l}
dP_-(t)=-\left[ 2r(t)P_-(t)+H^*_-(t,P_-(t),\Lambda_-(t)) \right]dt+\Lambda'_-(t)dW(t),\\
P_-(T)=a,\\
P_-(t)>0, \ \forall t\in[0,T].
\end{array}
\right.
\end{equation}
where 
\begin{align*}
H_-(t,v,P,\Lambda):=& v' P\sigma(t)\sigma'(t)v - 2v'\left[ \sigma(t)\theta(t)P+\sigma(t)\Lambda \right],\\
H^*_-(t,P,\Lambda):=& \inf_{v\in K}H_-(t,v,P,\Lambda).
\end{align*}

We find that \eqref{quad_min_con_dual_BSDE_P+} and \eqref{quad_min_con_dual_BSDE_P-} are  the extended SRE introduced in \cite{huzhou:constrainedmv}. Through the dual approach, we have obtained an explicit representation of the unique solution to the SREs in terms of the optimal state and adjoint processes. Finally, according to Theorem \ref{dual_to_primal} we conclude that the optimal solution to the primal problem is given by
\begin{equation*}
\left\{
\begin{array}{l}
\hat{\pi}'(t)=[\sigma']^{-1}(t)\hat{q}_2(t),\\
\hat{X}(t)=\hat{p}_2(t)=-\hat Y(t)\left[\dfrac{1_{\{x_0>0\}}}{P_+(t)}+\dfrac{1_{\{x_0<0\}}}{P_-(t)}\right].
\end{array}
\right.
\end{equation*}

\begin{remark}
If $ K=\mathbb{R}^N $, then we must have the optimal control $\hat\beta(t)=0$, which leads to $\hat\gamma(t)=0$ and $\hat H(t)=\Gamma(t)$  for $ t\in[0,T] $ a.e.  Condition (\ref{quad_min_con_dual_BSDE_eqn1}) is equivalent to
$$ \hat p_1(t)\theta(t)+\hat q_1(t)=0.$$
Replacing $\hat p_1(t)$ and $\hat q_1(t)$ by $P_+(t)$, $\Lambda_+(t)$ and $\hat\xi_+(t)$, we have
$$ \sigma\rq{}(t)\hat\xi_+(t)+ \theta(t) + {\Lambda_+(t) \over P_+(t)} = 0.$$
BSDE (\ref{quad_min_con_eqn1}) (or (\ref{quad_min_con_dual_BSDE_P+})) becomes
\begin{align*}
dP_+(t)=\left[ -2r(t)P_+(t)+2\theta'(t)\Lambda_+(t)+\theta'(t)\theta(t)P_+(t)+\dfrac{\Lambda'_+(t)\Lambda_+(t)}{P_+(t)} \right]dt+\Lambda'_+(t)dW(t),
\end{align*}
which is  the SRE introduced in \cite{zhoulim:MVrandom}. Using the duality approach, we obtain an explicit representation of the unique solution to the SRE. 
\end{remark}

\subsection{Deterministic coefficient case} \label{sec4.2}
Assume $ K\subset \mathbb{R}^N $ is a closed convex cone and
 $ r,b,\sigma $ are deterministic functions and $ a>0 $ is a constant. In this case, the dual problem can be written as
$$\mbox{Minimize }
 x_0y+E\left[\dfrac{Y(T)^2}{2a}\right]$$
over $ (y,\beta)\in\mathbb{R}\times {\cal H}^2(0,T; \mathbb{R}^N)$  and $ Y $ satisfies the SDE \eqref{dual_fsde} with $\alpha(t)=0$ and $\beta(t)\in K^0$ for $  t\in[0,T] $ a.e., 
where $ K^0:= \{ \beta:\beta'\pi \leq 0, \forall \pi\in K\}$, the polar cone of $K$. We solve the above problem in  two steps: first,  fix $ y $ and find the optimal control $ \hat{\beta}(y) $; second,  find the optimal $ \hat{y} $. We can then construct the optimal solution explicitly.

\medskip\noindent\textbf{Step 1:} Consider the associated HJB equation:
\begin{equation}\label{quad_min_con_det_dual_bellman}
\left\{
\begin{array}{l}
v_t(s,y)-r(s)yv_y(s,y)+\frac{1}{2}\inf_{\beta\in K^0}|\sigma^{-1}(s)\beta-\theta(s)y|^2v_{yy}(s,y)=0,\vspace{2mm}\\
v(T,y)=y^2,
\end{array}
\right.
\end{equation}
for each $ (s,y)\in [t,T]\times\mathbb{R} $. The infimum term in \eqref{quad_min_con_det_dual_bellman} can be written explicitly as 
\begin{enumerate}
\item If $ y=0 $, then it is trivial to obtain that
\begin{align*}
\inf_{\beta\in K^0}| \sigma^{-1}(s)\beta-\theta(s)y |^2=\inf_{\beta\in K^0}| \sigma^{-1}(s)\beta|^2 = 0.
\end{align*}
\item If $ y> 0 $, then we have
\begin{eqnarray*}
\inf_{\beta\in K^0}| \sigma^{-1}(s)\beta-\theta(s)y |^2
&=& y^2 \inf_{\beta\in K^0}\bigg| \sigma^{-1}(s)\left( \dfrac{\beta}{y}\right)-\theta(s) \bigg|^2\\
&=& y^2 \inf_{y\bar{\beta}\in K^0}\bigg| \sigma^{-1}(s)\bar{\beta}-\theta(s) \bigg|^2\\
&=& y^2 | \sigma^{-1}(s)\beta_+(s)-\theta(s) |^2,
\end{eqnarray*}
where $ \beta_+(s):= \argmin_{\beta\in K^0} | \sigma^{-1}(s)\beta-\theta(s) |^2 $.
\item If $ y< 0 $, then similarly we have
\begin{align*}
\inf_{\beta\in K^0}| \sigma^{-1}(s)\beta-\theta(s)y |^2=&y^2 \inf_{\beta\in K^0}\bigg| \sigma^{-1}(s)\dfrac{\beta}{y}-\theta(s) \bigg|^2 \\
=&y^2 \inf_{\bar{\beta}\in K^0}\bigg| \sigma^{-1}(s)\bar{\beta}+\theta(s) \bigg|^2\\
=& y^2 | \sigma^{-1}(s)\beta_-(s)+\theta(s) |^2,
\end{align*}
where $ \beta_-(s):= \argmin_{\beta\in K^0} | \sigma^{-1}(s)\beta+\theta(s) |^2 $.
\end{enumerate}
Define
\begin{equation*}
\sigma(s,y) :=
\left\{
\begin{array}{lcl}
\sigma^{-1}(s)\beta_+(s)-\theta(s), &\textit{if } y>0\\
\sigma^{-1}(s)\beta_-(s)+\theta(s), &\textit{if } y<0\\
0, &\textit{if } y=0.
\end{array}
\right.
\end{equation*}
The HJB equation \eqref{quad_min_con_det_dual_bellman} becomes
\begin{equation*}
\left\{
\begin{array}{l}
v_t(s,y)-r(s)yv_y(s,y)+\frac{1}{2}y^2|\sigma(s,y)|^2v_{yy}(s,y)=0,\\
v(T,y)=y^2.
\end{array}
\right.
\end{equation*}
According to the Feynman-Kac formula, we have
\begin{equation*}
v(t,y)=E\left[ Y^2(T)|Y(t)=y \right]=y^2e^{\int_t^T [-2r(s)+ |\sigma(s,Y(s))|^2]ds},
\end{equation*}
where the stochastic process $ Y $ follows the following geometric Brownian motion
\begin{equation*}
dY(s)=-r(s)Y(s)ds+\sigma'(s,Y(s))Y(s)dW(s), \ Y(t)=y.
\end{equation*}
Moreover, since $ Y $ follows a geometric Brownian motion and $ \sign(Y(s))=\sign(y), \ \forall s\in[t,T] $, we have
\begin{equation*}
\sigma(s,Y(s))=\sigma(s,y), \ \forall s\in[t,T].
\end{equation*}
In particular, we have
\begin{equation}\label{quad_min_con_det_dual_step1}
v(0,y)=y^2e^{\int_0^T [-2r(s)+ |\sigma(s,y)|^2]ds}.
\end{equation}

\medskip\noindent\textbf{Step 2:} Consider the following static optimization problem: 
\begin{equation}\label{quad_min_con_det_dual_step2}
\inf_{y\in\mathbb{R}}x_0y+\frac{1}{2a}v(0,y)
\end{equation}
Substituting \eqref{quad_min_con_det_dual_step1} into the objective function, we obtain that problem \eqref{quad_min_con_det_dual_step2} achieves minimum at
\begin{equation*}
\hat{y}=-ax_0e^{\int_0^T [2r(s)-|\sigma(s,-x_0)|^2]ds}.
\end{equation*}
Hence, we conclude that the optimal control is given by
\begin{equation*}
\hat{\beta}(t)=
\left\{
\begin{array}{lc}
ax_0e^{\int_t^T [2r(s)-|\sigma(s,-x_0)|^2]ds}\beta_-(t), &\textit{ if } x_0>0 \vspace{2mm}\\
-ax_0e^{\int_t^T [2r(s)-|\sigma(s,-x_0)|^2]ds}\beta_+(t), &\textit{ if } x_0<0 \vspace{2mm}\\
0, &\textit{ if } x_0=0.
\end{array}
\right.
\end{equation*}
Using the dual optimal control $(\hat y, \hat\beta)$, we can find a solution $(\hat Y, \hat p_2, \hat q_2)$ to the dual FBSDE 
(\ref{dual_FBSDE}) and (\ref{dual_conditions}), and then apply Theorem \ref{dual_to_primal} to  construct a solution $(\hat X, \hat p_1, \hat q_1)$ to the primal FBSDE 
(\ref{primal_necesuff_FBSDE}) and (\ref{primal_condition}).   
Moreover, in this case we can construct a solution to the SREs \eqref{quad_min_con_dual_BSDE_P+} and \eqref{quad_min_con_dual_BSDE_P-} explicitly as 
\begin{equation}\label{explicit_SRE}
\hat{P}_+(t)=\hat{P}_-(t)= ae^{\int_t^T [2r(s)+\sigma'(s,-x_0)\theta(s) ]ds}.
\end{equation}
Next, we verify that \eqref{explicit_SRE} are indeed solutions to the SREs \eqref{quad_min_con_dual_BSDE_P+} and \eqref{quad_min_con_dual_BSDE_P-} with $\Lambda_+(t)=0$ and $\Lambda_-(t)=0$, respectively. To this end, we consider the case $ x_0>0 $ and $ y<0 $. According to Theorem \ref{dual_to_primal}, we have
\begin{equation*}
\hat{X}(t)=\hat{p}_2(t), \forall t\in[0,T], a.e.
\end{equation*}
Hence, 
\begin{equation}\label{x_bsde}
\hat{X}(t)=E\left[ -\dfrac{\Gamma(T)Y(T)}{a\Gamma(t)}\bigg|\mathcal{F}_t \right]=-\dfrac{Y(t)}{a}E\left[\dfrac{\Gamma(T)Y(T)}{\Gamma(t)Y(t)}\bigg|\mathcal{F}_t \right],
\end{equation}
where $ \Gamma $ follows the SDE
\begin{equation*}
d\Gamma(t)=\Gamma(t)[-r(t)dt-\theta'(t)dW(t)], \forall t\in[0,T], \Gamma(0)=1.
\end{equation*}
Applying Ito's lemma, we obtain
\begin{equation}\label{GammaY}
d\Gamma(t)Y(t)=[-2r(t)-\theta'(t)\sigma(t,y)]Y(t)\Gamma(t)dt-[\sigma'(t,y)+\theta'(t)]Y(t)\Gamma(t)dW(t).
\end{equation}
Combining \eqref{x_bsde} and \eqref{GammaY},  we have
$$ \hat X(t) = -{Y(t)\over a} e^{\int_t^T [-2r(s)-\theta\rq{}(s)\sigma(s,y)] ds}.$$
Applying Ito's lemma again, we have  $\hat X$ satisfies the SDE
\begin{equation}\label{hat_x_explicit}
d\hat{X}(t)=[r(t)\hat{X}(t)+\theta'(t)\sigma(t,y)\hat{X}(t)]dt+\sigma'(t,y)\hat{X}(t)dW(t).
\end{equation}
Comparing \eqref{hat_x_explicit} with \eqref{wealth_process}, we conclude that
\begin{equation*}
\hat{\pi}'(t)=\sigma'(t,y)\sigma^{-1}(t)\hat{X}(t),
\end{equation*}
which implies that
\begin{equation}\label{optimal_xi_explicit}
\hat{\xi}_+'(t)=\dfrac{\hat{\pi}'(t)}{\hat{X}(t)}=\sigma'(t,y)\sigma^{-1}(t).
\end{equation}
Substituting \eqref{optimal_xi_explicit} back into \eqref{quad_min_con_eqn2}, we have
\begin{equation*}
H^*(t,P_+(t),\Lambda_+(t))=\sigma'(t,y)\theta(t)P_+(t).
\end{equation*}
Taking $ x_0<0 $ and following the same steps, we obtain
\begin{equation*}
H^*(t,P_-(t),\Lambda_-(t))=\sigma'(t,y)\theta(t)P_-(t).
\end{equation*}
Hence, we conclude that $\hat P_+(t)$ and $\hat P_-(t)$ defined in 
   \eqref{explicit_SRE} are indeed solutions to SREs \eqref{quad_min_con_dual_BSDE_P+} and \eqref{quad_min_con_dual_BSDE_P-}.

\section{Proofs of the Main Results}\label{section5}
In this section we give proofs of the main results in Section \ref{mainresults}.

\bigskip\noindent
{\it Proof of Theorem \ref{primal_necesuff}}. 
Since the cost functional $ J$ is convex, according to \cite[Proposition 2.2.1]{ekelandtemam:convexanalysis}, a necessary and sufficient condition for $ \hat{\pi} $ to be optimal  is that
\begin{equation}\label{primal_smp_gat_1}
\langle J'(\hat{\pi}),\hat{\pi}-\pi \rangle \leq 0, \; \forall \pi\in\mathcal{A},
\end{equation}
where $J'(\hat{\pi})$ is the G\^{a}teaux-dirivative of $J$ at $\hat\pi$ and can be computed explicitly as \eqref{wealth_process} is a linear SDE  and $J$ is a quadratic functional. The optimality condition (\ref{primal_smp_gat_1}) can be written as
\begin{align}\nonumber
E\bigg[\displaystyle\int_0^T\bigg[ &Q(t)X^{\hat{\pi}}(t)\left(X^{\hat{\pi}}(t)-X^{\pi}(t)\right)+S'(t)\left(\hat\pi(t)\left(X^{\hat{\pi}}(t)-X^{\pi}(t)\right)
+\left(\hat{\pi}(t)-\pi(t)\right) X^{\hat\pi}(t)\right)\\\label{primal_smp_gat_2}
&+\left(\hat{\pi}'(t)-\pi'(t)\right)R(t)\hat{\pi}(t)\bigg]dt
+\left[ aX^{\hat\pi}(T) + c\right] \left(X^{\hat{\pi}}(T)-X^{\pi}(T)\right)\bigg]\leq 0,
\end{align} 
for all $ \pi\in\mathcal{A} $. Applying Ito's formula to $ X^{\hat{\pi}}(t)\hat{p}_1(t) $, we have
\begin{flalign}\notag
d(X^{\hat{\pi}}(t)\hat{p}_1(t))
=&\left[ \hat{p}_1(t)\hat{\pi}'(t)\sigma(t)\theta(t) + \hat{\pi}'(t)\sigma(t)\hat{q}_1(t)+Q(t)X^{\hat{\pi}}(t)^2+S'(t)X^{\hat{\pi}}(t)\hat{\pi}(t) \right]dt \\\label{primal_smp_ito_1}
&+ \left[ \hat{p}_1(t)\hat{\pi}'(t)\sigma(t)+\hat{q}_1'(t)X^{\hat{\pi}}(t) \right] dW(t).
\end{flalign}
Define the process $ \tilde S $ as
\begin{equation*}
\tilde S(t):= \int_0^t \left(\hat{p}_1(s)\hat{\pi}'(s)\sigma(s)+\hat{q}'_1(s)X^{\hat{\pi}}(s)\right)dW(s),\; 0\leq t\leq T.
\end{equation*}
Obviously, $\tilde S $ is a local martingale. To prove that $\tilde S $ is  a true martingale, it is sufficient to show that $ E\left[ \sup_{0\leq s\leq T}|\tilde S(s)| \right]<\infty $. According to the Burkholder-Davis-Gundy inequality \cite[Theorem 3.3.28]{karatzasshreve:stochasticcalculus}, it is sufficient to verify that
\begin{equation*}
E\left[ \left( \int_0^T [| \hat{p}_1(s)\pi'(s)\sigma(s) |^2+| \hat{q}_1(s)X^{\hat{\pi}}(s) |^2 ] ds \right)^{\frac{1}{2}} \right]<\infty.
\end{equation*}
Note that from \cite[Corollary 2.5.10]{kylovcontrolled}, we have that
$X^{\hat{\pi}}\in {\cal S}^2(0,T; \mathbb{R})$. 
Combining with $ p_1\in {\cal S}^2(0,T; \mathbb{R})$ and $q_1\in {\cal H}^2(0,t;\mathbb{R}^N)$ and by H\"{o}der's inequality, we have
\begin{align*}
&E\left[ \left( \int_0^T [ | \hat{p}_1(s)\pi'(s)\sigma(s) |^2+| \hat{q}_1(s)X^{\hat{\pi}}(s) |^2 ] ds \right)^{\frac{1}{2}} \right]\\
&\leq E\left[ \left( \sup_{0\leq s \leq T}|\hat{p}_1(s)|^2\int_0^T |\pi'(s)\sigma(s)|^2 ds+\sup_{0\leq s\leq T}|X^{\hat{\pi}}(s)|^2\int_0^T | q_1(s) |^2ds \right)^{\frac{1}{2}} \right]\\
&\leq \frac{1}{2}E\left[ \sup_{0\leq s\leq T}|\hat{p}_1(s)|^2 \right]+\frac{1}{2}E\left[ \int_0^T |\pi'(s)\sigma(s)|^2 ds \right] +\frac{1}{2}E\left[\sup_{0\leq s\leq T}|X^{\hat{\pi}}(s)|^2\right]+\frac{1}{2}E\left[\int_0^T|\hat{q}_1(s)|^2ds\right]\\
&<\infty,
\end{align*}
which implies that $\tilde S $ is  a true martingale. Taking expectation of $ X^{\hat{\pi}}(T)\hat{p}_1(T) $, we have
\begin{align}\label{primal_smp_ito_2}
E\left[ X^{\hat{\pi}}(T)\hat{p}_1(T) \right] = x_0\hat p_1(0) + E \bigg[ \int_0^T &\bigg[ \hat{p}_1(t)\hat{\pi}'(t)\sigma(t)\theta(t) + \hat{\pi}'(t)\sigma(t)\hat{q}_1(t)\\\notag
&+Q(t)X^{\hat{\pi}}(t)^2+S'(t)X^{\hat{\pi}}(t)\hat{\pi}(t)\bigg]dt \bigg]
\end{align}
Similarly, applying Ito's formula to $ X^{\pi}(t)\hat p_1(t) $ and taking expectation, we obtain that
\begin{align}\label{primal_smp_ito_3}
E\left[ X^{\pi}(T)\hat{p}_1(T) \right] = x_0\hat p_1(0) + E \bigg[ \int_0^T & \bigg[ \hat{p}_1(t)\pi'(t)\sigma(t)\theta(t) + \pi'(t)\sigma(t)\hat{q}_1(t)\\\notag
&+Q(t)X^{\hat{\pi}}(t)X^{\pi}(t)+S'(t)X^{\pi}(t)\hat{\pi}(t) \bigg]dt \bigg]
\end{align}
Combining \eqref{primal_smp_gat_2},\eqref{primal_smp_ito_2} and \eqref{primal_smp_ito_3}, we obtain that $ \hat{\pi}\in\mathcal{A} $ is an optimal control of the primal problem if and only if 
\begin{equation}\label{primal_smp_Ham_1}
E\left[ \int_0^T \left[ \hat{\pi}'(t)-\pi'(t) \right]\left[ \hat{p_1}(t)\sigma(t)\theta(t)+\sigma(t)\hat{q}_1(t)+S(t)X^{\hat{\pi}}(t)+R(t)\hat{\pi}(t) \right] dt \right] \geq 0
\end{equation}
for all $\pi\in\mathcal{A} $.
Define the Hamiltonian function $ H:\Omega\times [0,T]\times \mathbb{R} \times  \mathbb{R}^N \rightarrow \mathbb{R} $ as
\begin{equation*}
H(\omega, t,x,\pi):= \pi'\left[ \hat{p_1}(t)\sigma(t)\theta(t)+\sigma(t)\hat{q}_1(t)+S(t)x+\dfrac{1}{2}R(t)\pi \right]
\end{equation*}
and define the set-valued map $ F: \Omega\times [0,T] \rightarrow K $ as 
\begin{equation*}
F(\omega,t):= \left\{ \pi\in K:  \left[\hat{\pi}'(t)-\pi'\right]H_{\pi}\left(\omega, t,X^{\hat{\pi}}(t),\hat{\pi}(t) \right) \geq 0 \right\}.
\end{equation*}
Then $ F $ is a measurable set-valued map, see \cite[Definition 8.1.1]{aubin:setvaluedanalysis}.
Given $ \pi\in K $, define the set $ \mathbb{B}^{\pi} $ as
\begin{align*}
\mathbb{B}^{\pi} &:= \left\{ (\omega,t)\in\Omega\times [0,T] : \left[ \hat{\pi}'(t)-\pi' \right]H_{\pi}(t,X^{\hat{\pi}}(t),\hat{\pi}(t))< 0 \right\}.
\end{align*}
According to   \cite[Theorem 8.14]{aubin:setvaluedanalysis}, $\mathbb{B}^{\pi}_t\in\mathcal{F}_t $ for  $ t\in[0,T]$. Define an adapted control $ \tilde{\pi}: \Omega\times [0,T] \rightarrow K $ as
\begin{equation*}
\tilde{\pi}(\omega,t):=
\left\{
\begin{array}{l}
\pi \hspace{16mm} \textit{            if }(\omega,t)\in\mathbb{B}^\pi\\
\hat{\pi}(\omega,t),\hspace{5mm} \textit{            otherwise.}
\end{array}
\right.
\end{equation*}
Suppose that $ \left(\mathbb{P}\otimes Leb\right)(\mathbb{B}^{\pi})>0 $, then
\begin{equation*}
E\left[ \int_0^T [\hat{\pi}'(t)-\tilde{\pi}'(t)]H_{\pi}(t,X^{\hat{\pi}}(t),\hat{\pi}(t))dt \right]< 0,
\end{equation*}
contradicting with \eqref{primal_smp_Ham_1}. Hence, we conclude that $ \left(\mathbb{P}\otimes Leb\right)(\mathbb{B}^{\pi})=0 $ for any fixed  $  \pi\in K $. Moreover, since $ K $ is separable, we conclude that 
\begin{equation*}
[\hat{\pi}'(t)-\pi']H_{\pi}(t,X^{\hat{\pi}}(t),\hat{\pi}(t))\geq 0, \; \forall \pi\in K
\end{equation*}
for $\left(\mathbb{P}\otimes Leb\right)$-a.e. $ (\omega,t)\in \Omega\times [0,T] $.
\hfill $\Box$

\bigskip\noindent
{\it Proof of Theorem \ref{Dual_Nece_Cond}}. 
Let $ (\hat{y},\hat{\alpha},\hat{\beta}) $ be  optimal for the dual problem and $ (Y^{(\hat{y},\hat{\alpha},\hat{\beta})}, 
\hat{p}_2,\hat{q}_2) $ satisfy  \eqref{dual_FBSDE}. 
Let $(y,\alpha,\beta)\in \mathbb{B}$ and $Y^{(y,\alpha,\beta)}$ satisfy the SDE (\ref{dual_fsde}). 
Applying Ito's formula to $ \hat{p}_2(t)Y^{(y,\alpha,\beta)}(t) $, we have
\begin{flalign*}
d(\hat{p}_2(t)Y^{(y,\alpha,\beta)}(t)) 
=&\left[{\alpha}(t)\hat{p}_2(t)+\hat{q}_2'(t)\sigma^{-1}(t){\beta}(t)\right]dt \\
&{} + \left[ \hat{q}_2'(t)Y^{(y,\alpha,\beta)}(t) + \left(\sigma^{-1}(t){\beta}(t)-\theta(t)Y^{(y,\alpha,\beta)}(t) \right)'\hat{p}_2(t) \right]dW(t).
\end{flalign*}
It can be shown, following a similar argument as in  the proof of Theorem \ref{primal_necesuff}, that
 the process 
\begin{equation*}
 \int_0^t \left[  \hat{q}_2'(s)Y^{(y,\alpha,\beta)}(s)+( \sigma^{-1}(s){\beta}(s)-\theta(s)Y^{(y,\alpha,\beta)}(s) )'\hat{p}_2(s) \right]dW(s), \ 0\leq t\leq T, 
\end{equation*}
is a martingale. 
Taking the expectation of $ \hat{p}_2(T)Y^{(y,\alpha,\beta)}(T) $, we obtain
\begin{equation}\label{nece_dual_proof_adjoint_exp}
E\left[ \hat{p}_2(T)Y^{(y,\alpha,\beta)}(T) \right] = \hat{p}_2(0){y} + E\left[ \int_0^T \left[{\alpha}(t)\hat{p}_2(t)+\hat{q}_2'(t)\sigma^{-1}(t){\beta}(t)\right] dt \right].
\end{equation}
For $ \varepsilon>0$ define $ (y^{\varepsilon},\alpha^{\varepsilon},\beta^{\varepsilon})
\in \mathbb{B}$ by
\begin{equation*}
 (y^{\varepsilon},\alpha^{\varepsilon},\beta^{\varepsilon})
= (\hat{y},\hat{\alpha},\hat{\beta}) + \varepsilon (y,\alpha,\beta).
\end{equation*}
Then
\begin{equation*}
Y^{(y^{\varepsilon},\alpha^{\varepsilon},\beta^{\varepsilon})}(t) = Y^{(\hat{y},\hat{\alpha},\hat{\beta})}(t) + \varepsilon Y^{(y,\alpha,\beta)}(t).
\end{equation*}
Since $ (\hat{y},\hat{\alpha},\hat{\beta}) $ is optimal, we have
\begin{equation*}
\dfrac{1}{\varepsilon} \left[ \Psi(y^{\varepsilon},\alpha^{\varepsilon},\beta^{\varepsilon})-\Psi(\hat{y},\hat{\alpha},\hat{\beta}) \right] \geq 0.
\end{equation*}
Substituting (\ref{dual_cost_function}) into the above inequality, also noting $\hat{p}_2(T)=-\dfrac{Y^{(\hat{y},\hat{\alpha},\hat{\beta})}(T) + c}{a} $, we get 
\begin{equation}\label{necedual_proof_2}
\begin{array}{c}
y x_0 -E\left[ Y^{(y,\alpha,\beta)}(T)\hat{p}_2(T) \right] + \varepsilon E\left[ \dfrac{ Y^{(y,\alpha,\beta)}(T)^2}{2a} \right] \\
{} + \dfrac{1}{\varepsilon}E\left[ \displaystyle\int_0^T \left[\phi(\alpha^{\varepsilon}(t),\beta^{\varepsilon}(t))-\phi(\hat{\alpha}(t),\hat{\beta}(t))\right] dt \right] \geq 0.
\end{array}
\end{equation}
Combining  \eqref{necedual_proof_2} with \eqref{nece_dual_proof_adjoint_exp} and then letting $\varepsilon\downarrow 0 $, we have
$$
y\left( x_0-\hat{p}_2(0) \right)+\lim_{\varepsilon\downarrow 0}E\left[ \displaystyle\int_0^T [ \tilde g(t,\varepsilon) - \hat{q}_2'(t)\sigma^{-1}(t)\beta(t)-\alpha(t)\hat{p}_2(t)]dt \right] \geq 0,
$$
where $\tilde g(\omega,t,\varepsilon)=\dfrac{1}{\varepsilon}(\phi(t,\alpha^{\varepsilon}(t),\beta^{\varepsilon}(t))-\phi(t,\hat{\alpha}(t),\hat{\beta}(t)))$. 
Let $ \alpha(t)=0 $ and $\beta(t)=0 $ for $ t\in[0,T] $, we get
\begin{equation*}
y(x_0-\hat{p}_2(0))\geq 0, \ \forall y\in\mathbb{R}.
\end{equation*}
Hence,  $ \hat{p}_2(0)=x_0 $. 
Recall that the function $ f $ in \eqref{cost_functions_f_g} is convex and the  set $K$ is convex, according to  \cite[Theorem 26.3]{rockafeller:convexanalysis}, $ \phi $ has directional derivative at  $(\hat{\alpha}(t),\hat{\beta}(t))$ in any direction 
$\left(\mathbb{P}\otimes Leb\right)$ a.e. on $\Omega\times [0.T]$. Since $\varepsilon\to \tilde g(\omega,t,\varepsilon)$ is a nondecreasing function, Assumption \ref{assum_monoton} and the monotone convergence theorem imply that 
\begin{equation}\label{necedual_proof_5}
E\left[ \int_0^T \left[ \phi^{o}\left(t,\hat{\alpha}(t),\hat{\beta}(t);{\alpha}(t),{\beta}(t)\right)-\hat{q}_2'(t)\sigma^{-1}(t){\beta}(t)-{\alpha}(t)\hat{p}_2(t) \right]dt \right] \geq 0，
\end{equation}
where
\begin{equation*}
\phi^{o}\left(\omega,t,\hat{\alpha},\hat{\beta};{\alpha},{\beta}\right):= \lim_{\varepsilon\downarrow 0}\dfrac{\phi(t,\hat{\alpha}+\varepsilon{\alpha},\hat{\beta}+\varepsilon{\beta})-\phi( t,\hat{\alpha},\hat{\beta})}{\varepsilon}.
\end{equation*}
For $ (\alpha,\beta)\in\mathbb{R}\times\mathbb{R}^N $, define the set $ \mathbb{B}^{(\alpha,\beta)} $ as
\begin{equation*}
\mathbb{B}^{(\alpha,\beta)}:=\left\{ (\omega,t)\in \Omega\times [0,T]: \phi^{o}\left(\hat{\alpha}(t),\hat{\beta}(t);\alpha,\beta\right)-\hat{q}_2'(t)\sigma^{-1}(t)\beta-\alpha\hat{p}_2(t)<0 \right\}.
\end{equation*}
Using a similar argument as in the proof of Theorem \ref{primal_necesuff}, we conclude that $\mathbb{B}^{(\alpha,\beta)}_t \in \mathcal{F}_t $ for $ t\in[0,T]$  and 
 $\left(\mathbb{P}\otimes Leb\right)(\mathbb{B}^{(\alpha,\beta)})=0 $ for all $(\alpha,\beta)\in\mathbb{R}\times\mathbb{R}^N $. Equivalently, given any $ (\alpha,\beta)\in\mathbb{R}\times\mathbb{R}^N $,
\begin{equation*}
\phi^{o}\left(\hat{\alpha}(t),\hat{\beta}(t);\alpha,\beta\right)-\hat{q}_2'(t)\sigma^{-1}(t)\beta-\alpha\hat{p}_2(t)\geq 0,
\end{equation*}
for $ \left(\mathbb{P}\otimes Leb\right) $-a.e. $ (\omega,t)\in \Omega\times [0,T] $. In addition, by the separability of the space $ \mathbb{R}^{N+1} $, we conclude that
\begin{equation*}
\phi^{o}\left(\hat{\alpha}(t),\hat{\beta}(t);\alpha,\beta\right)-\hat{q}_2'(t)\sigma^{-1}(t)\beta-\alpha\hat{p}_2(t)\geq 0, \forall (\alpha,\beta)\in\mathbb{R}\times\mathbb{R}^N
\end{equation*}
for $ \left(\mathbb{P}\otimes Leb\right) $-a.e. $ (\omega,t)\in \Omega\times [0,T] $. By the definition of Clarke's generalized gradient \cite[Chapter 2]{clarke:optimization}, the above condition can be written as
\begin{equation*}
\left( \hat{p}_2(t),[\sigma']^{-1}(t)\hat{q}_2(t) \right)\in\partial\phi\left( \hat{\alpha}(t),\hat{\beta}(t) \right).
\end{equation*}
According to  \cite[Theorem 23.5]{rockafeller:convexanalysis} , we conclude that $ x\hat{\alpha}(t)+\pi'\hat{\beta}(t)-\tilde f(t,x,\pi) $ achieves the supreme at $ (\hat{p}_2(t),[\sigma']^{-1}(t)\hat{q}_2(t)) $
for $ \left(\mathbb{P}\otimes Leb\right) $-a.e. $ (\omega,t)\in \Omega\times [0,T] $, which implies  
\begin{equation*}
[\sigma']^{-1}(t)\hat{q}_2(t)\in K,
\end{equation*}
for $ \left(\mathbb{P}\otimes Leb\right) $-a.e. $ (\omega,t)\in \Omega\times [0,T] $.
We have proved the necessary condition.

Let $ (\hat{y},\hat{\alpha},\hat{\beta})\in\mathbb{B}$ be an admissible control to the dual problem with processes $ \left( Y^{(\hat{y},\hat{\alpha},\hat{\beta})},\hat{p}_2,\hat{q}_2 \right) $ satisfying  the FBSDE \eqref{dual_FBSDE} and conditions \eqref{dual_conditions}. Define the Hamiltonian function $ H:\Omega\times [0,T]\times\mathbb{R}\times\mathbb{R}^N \rightarrow \mathbb{R} $ as 
\begin{equation*}
H(\omega, t,\alpha,\beta)=\hat{q}_2'(t)\sigma^{-1}(t)\beta+\alpha\hat{p}_2(t)-\phi(t,\alpha,\beta).
\end{equation*}
By condition \eqref{dual_conditions} and the classical result in duality theorem, we have
\begin{equation}\label{dual_H_condition}
\left(0,0\right)\in\partial H\left( \hat{\alpha}(t),\hat{\beta}(t) \right),
\end{equation}
for $ \left(\mathbb{P}\otimes Leb\right) $-a.e. $ (\omega,t)\in \Omega\times [0,T] $. Given any admissible control $( y,\alpha,\beta )\in\mathbb{B}$, define 
\begin{equation*}
\tilde{y}=y-\hat{y}, \ \tilde{\alpha}=\alpha-\hat{\alpha}, \ \tilde{\beta}=\beta-\hat{\beta}.
\end{equation*}
Let $ Y^{( y,\alpha,\beta )} $ and $ Y^{(\tilde{y},\tilde{\alpha},\tilde{\beta})} $ be the associated state processes satisfying the SDE \eqref{dual_fsde}. According to the definition of the dual problem, also noting $m_T$ is a convex function, we have
\begin{align*}
\Psi\left( y,\alpha,\beta \right)-\Psi\left( \hat{y},\hat{\alpha},\hat{\beta} \right)
\geq & \tilde{y}x_0+E\left[ Y^{(\tilde{y},\tilde{\alpha},\tilde{\beta})}(T)\dfrac{Y^{(\hat{y},\hat{\alpha},\hat{\beta})}(T)+c}{a} \right]\\
&{}+E\left[ \displaystyle\int_0^T\left[ \phi(t,\alpha(t),\beta(t))-\phi(t,\hat{\alpha}(t),\hat{\beta}(t)) \right]dt \right].
\end{align*}
Replacing $ \dfrac{Y^{(\hat{y},\hat{\alpha},\hat{\beta})}+c}{a} $ with $ -\hat{p}_2(T) $ in the above inequality, we have
\begin{align*}
\Psi\left( y,\alpha,\beta \right)-\Psi\left( \hat{y},\hat{\alpha},\hat{\beta} \right) \geq & \tilde{y}(x_0-\hat{p}_2(0))+E\left[ \int_0^T\left[ \hat{q}_2'(t)\sigma^{-1}(t)\tilde{\beta}(t)-\tilde{\alpha}(t)\hat{q}_2(t) \right]dt \right]\\
&{} +E\bigg[ \displaystyle\int_0^T\left[ \phi(t,\alpha(t),\beta(t))-\phi(t,\hat{\alpha}(t),\hat{\beta}(t)) \right]dt \bigg]\\
=& E\left[ \int_0^T\left[ -H(t,\alpha(t),\beta(t))+H(t,\hat{\alpha}(t),\hat{\beta}(t)) \right] dt  \right].
\end{align*}
According to condition \eqref{dual_H_condition} and the concavity of $ H $, we conclude that
\begin{equation*}
\Psi\left(\bar{y},\bar{\alpha},\bar{\beta}\right)-\Psi\left( \hat{y},\hat{\alpha},\hat{\beta} \right)\geq 0.
\end{equation*}
Since $ \left( y,\alpha,\beta \right)\in\mathbb{B} $ is arbitrary, we have proved the sufficient condition. \hfill$\Box$

\bigskip\noindent
{\it Proof of Theorem \ref{dual_to_primal}}. 
Suppose that $ (\hat{y},\hat{\alpha},\hat{\beta}) \in \mathbb{B}$   is optimal for the dual problem. By Theorem \ref{Dual_Nece_Cond}, the process $ \left( Y^{(\hat{y},\hat{\alpha},\hat{\beta})}(t),\hat{p}_2(t),\hat{q}_2(t) \right) $ solves the dual FBSDE (\ref{dual_FBSDE}) and satisfies condition \eqref{dual_conditions}. 
Define $\hat\pi(t)$ and $(X^{\hat{\pi}}(t), \hat{p}_1(t), \hat{q}_1(t))$ as in (\ref{relations_pi_defn}) and (\ref{dual_to_primal_soln}), respectively. 
According to Theorem \ref{Dual_Nece_Cond} and condition \eqref{dual_conditions}, we have $ \hat{\pi}(t)\in K \ \mathbb{P}$-a.s. and 
\begin{equation*}
 \left( X^{\hat{\pi}}(t),\hat{\pi}(t) \right)\in \partial\phi\left( \hat{\alpha}(t),\hat{\beta}(t) \right).
\end{equation*}
The classical result in duality theory implies
\begin{equation*}
\left( \hat{\alpha}(t),\hat{\beta}(t) \right)\in \partial \tilde f\left(X^{\hat{\pi}}(t),\hat{\pi}(t) \right).
\end{equation*}
Recall that $\tilde f(\omega,t,x,\pi)=f(\omega,t,x,\pi) + \Psi_K(\pi)$, we can get
\begin{eqnarray}\label{dual_to_primal_alpha}
\hat{\alpha}(t)&=&Q(t)X^{\hat{\pi}}(t)+S'(t)\hat{\pi}(t),\\
\label{dual_to_primal_beta}
\hat{\beta}(t)&\in& S(t)X^{\hat{\pi}}(t)+ R(t)\hat{\pi}(t) +  \partial \Phi_K( \hat{\pi}(t))
\label{dual_to_primal_beta}
\end{eqnarray}
for $ \left(\mathbb{P}\otimes Leb\right) $-a.e. $ (\omega,t)\in \Omega\times [0,T] $. Combining (\ref{relations_pi_defn}), \eqref{dual_to_primal_soln} and \eqref{dual_to_primal_alpha}, we obtain that $ \left( X^{\hat{\pi}},\hat{p}_1,\hat{q}_1 \right) $ solves the primal FBSDE (\ref{primal_necesuff_FBSDE}).
Moreover, combining  \eqref{dual_to_primal_soln} and \eqref{dual_to_primal_beta} gives condition \eqref{primal_condition}. Using the sufficient condition for optimality in Theorem \ref{primal_necesuff}, we conclude that $ \hat{\pi} $ is indeed an optimal control for the primal problem. \hfill$\Box$

\bigskip\noindent
{\it Proof of Theorem \ref{primaltodual}}. 
Suppose that $ \hat{\pi}\in\mathcal{A} $ is an optimal control for the primal problem. By Theorem \ref{primal_necesuff},  the process $ \left( X^{\hat{\pi}}(t),\hat{p}_1(t),\hat{q}_1(t) \right) $ solves the FBSDE (\ref{primal_necesuff_FBSDE}) and satisfies condition (\ref{primal_condition}). 
Define $(\hat{y}, \hat{\alpha}(t), \hat{\beta}(t))$ and $(Y^{(\hat{y},\hat{\alpha},\hat{\beta})}(t), \hat{p}_2(t), \hat{q}_2(t))$ as in (\ref{primal_to_dual_control}) and (\ref{primal_to_dual_cond}), respectively. 
Substituting them into the primal FBSDE  (\ref{primal_necesuff_FBSDE}), we obtain that $ \left( Y^{(\hat{y},\hat{\alpha},\hat{\beta})},\hat{p}_2,\hat{q}_2 \right) $ satisfies the dual FBSDE (\ref{dual_FBSDE}). 
By the construction in (\ref{primal_to_dual_control}) and  \eqref{primal_to_dual_cond}, the first two conditions in \eqref{dual_conditions} are satisfied. In addition, by condition \eqref{primal_condition} and the concavity of $ H $, we have 
\begin{equation*}
\hat{\beta}(t)\in\partial_{\pi}\tilde f\left(X^{\hat{\pi}}(t),\hat{\pi}(t)\right).
\end{equation*}
Consequently, we have
\begin{equation*}
\left(\hat{\alpha}(t),\hat{\beta}(t)\right)\in\partial \tilde f\left(X^{\hat{\pi}}(t),\hat{\pi}(t) \right),
\end{equation*}
which is equivalent to the third condition in \eqref{dual_conditions}. By Theorem   \ref{Dual_Nece_Cond}, we conclude that $ \left(\hat{y},\hat{\alpha},\hat{\beta}\right) $ is indeed an optimal control to the dual problem.
\hfill$\Box$

\section{Conclusion}
In this paper, we discuss a continuous-time constrained quadratic risk minimization problem  with random market coefficients. Following a convex duality approach, we derive the necessary and sufficient optimality conditions for primal and dual problems in terms of FBSDEs plus additional conditions. We establish an explicit connection between  primal and dual problems in terms of their associated forward backward systems. We prove that the optimal controls of primal and dual problems can be written as functions of adjoint processes of their counterpart. Moreover, we also find that the optimal state processes  for both problems coincide with the optimal adjoint processes of their counterpart. We solve cone-constrained quadratic risk minimization problems using the dual approach. We recover the solutions to the extended SREs introduced in the literature from the optimal solutions to the dual problem and find  the closed-form solutions to the extended SREs  when the coefficients are deterministic.
There are still many open questions. For example,   can the results be extended to incomplete market models (not a complete market model as in the paper)?  Can the dual problem be solved for a bounded control set $K$ (not a cone)?  Can solutions be found to the primal and dual FBSDEs with random  coefficients (not deterministic coefficients)?  We  leave these outstanding problems  in  future works.

\bigskip\noindent
{\bf Acknowledgments}. The authors are  grateful to  three anonymous reviewers whose constructive comments and suggestions have helped to improve the paper of the previous version.


\begin{thebibliography}{9}
\bibitem{aubin:setvaluedanalysis}
J-P. Aubin and H. Frankowska, \emph{Set-Valued Analysis}, Birkh\"{a}user, 1990.

\bibitem{bismut:convexdual} J. M. Bismut, \emph{Conjugate convex functions in optimal stochastic control}, J. Math. Anal. Appl.,  \textbf{44} (1973), pp. 384--404.

\bibitem{cadenillaskaratzas:convexsmp} A. Cadenillas and I. Karatzas, \emph{The stochastic maximum principle for linear convex optimal control with random coefficients}, SIAM J. Control Optim.,  \textbf{33} (1995), pp. 590--624.

\bibitem{Czichowsky:convexdualconstraints} C. Czichowsky and M. Schweizer, \emph{Convex duality in mean-variance hedging under convex trading constraints}, Adv. in Appl. Probab.,  \textbf{44} (2012), pp. 1084--1112.


\bibitem{cwz} C. Czichowsky,
N. Westray and  H Zheng, \emph{Convergence in the semimartingale topology and constrained portfolios},
Seminaire de Probabilities, \textbf{XLIII} (2011), pp. 395--412.


\bibitem{clarke:optimization}
F. H. Clarke, \emph{Optimization and Nonsmooth Analysis}, SIAM, 1990.

\bibitem{ekelandtemam:convexanalysis}
I. Ekeland and R. Tamam, \emph{Convex Analysis and Variational Problems}, SIAM, 1987.

\bibitem{huzhou:constrainedmv} Y. Hu and X. Y. Zhou, \emph{Constrained stochastic lq control with random coefficients, and application to portfolio selection}, SIAM J. Control Optim.,  \textbf{44} (2005), pp. 444--446.

\bibitem{karatzasshreve:stochasticcalculus}
I. Karatzas and S. E. Shreve, \emph{Brownian Motion and Stochastic Calculus}, Springer, 1998.

\bibitem{karatzasshreve:mathfinance}
I. Karatzas and S. E. Shreve, \emph{Methods of Mathematical Finance}, Springer, 2001.

\bibitem{KS99} D. Kramkov and W. Schachermayer, \emph{The asymptotic elasticity of utility functions and optimal investment in incomplete markets}, Ann. Appl. Probab., \textbf{9} (1999), pp. 904--950.

\bibitem{KS03} D. Kramkov and W. Schachermayer, \emph{Necessary and sufficient conditions in the problem of optimal investment in incomplete markets}, Ann. Appl. Probab., \textbf{13} (2003), pp. 1504--1516.

\bibitem{kylovcontrolled}
N. V. Krylov, \emph{Controlled Diffusion Processes}, Springer-Verlag, 1980.


\bibitem{heunislabbe:constrainedMV} C. Labb\'{e} and A. J. Heunis, \emph{Convex duality in constrained mean-variance portfolio optimization}, Adv. in Appl. Probab.,  \textbf{39} (2007), pp. 77--104.

\bibitem{oksendal:robustduality} B. {\O}ksendal and A. Sulem, \emph{A stochastic control approach to robust duality in utility maximization}, preprint (2013), available at \url{http://arxiv.org/abs/1304.5040}.

\bibitem{pham:contimuoustimeSC}
H. Pham, \emph{Continuous-time Stochastic Control and Optimization with Financial Applications}, Springer, 2009.

\bibitem{rockafeller:convexanalysis}
R. T. Rockafeller, \emph{Convex Analysis}, Princeton University Press, 1970.

\bibitem{rogers:constrainteddual} L. C. G. Rogers, \emph{Duality in constrained optimal investment and consumption problems: a synthesis}, in Paris-Princeton Lectures on Mathematical Finance, Springer, Berlin, 2002, pp. 95-131. 

\bibitem{schweizer:mvhedging}
M. Schweizer, \emph{Mean-variance hedging}, Encyclopedia of Quantitative Finance, 2010, pp. 1177-1181. 


\bibitem{yong2006} J. Yong, \emph{Linear forward-backward stochastic differential equations with random coefficients}, Probab. Theory Related Fields,  \textbf{135} (2006), pp. 53--83.

\bibitem{yong.zhou:stochasticcontrols}
J. Yong and X.Y. Zhou, \emph{Stochastic Controls: Hamiltonian Systems and HJB Equations}, Springer, 1999.

\bibitem{zhoulim:MVrandom} X. Y. Zhou and A. Lim, \emph{Mean-variance portfolio selection with random parameters in a complete market}, Math. Oper. Res.,  \textbf{27} (2002), pp. 101--120.

\end{thebibliography}

\end{document}